\documentclass[12pt,reqno]{article}

\usepackage{preamble_draft}

\begin{document}

\title{Optimal Auction Design for Dynamic Stochastic Environments:  Myerson Meets Naor\thanks{We are grateful to Omar Besbes, Simon Board, Tilman B\"orgers, Ben Brooks, Faruk Gul, Itai Gurvich, Navin Kartik,  Jinwoo Kim, Kyungmin Kim, Andreas Kleiner, Fuhito Kojima, Stephan Lauermann, Qingmin Liu, Will Ma, George Mailath, Pietro Ortoleva, Alessandro Pavan, Harry Pei, Justus Preusser, Doron Ravid, Philip Reny, Larry Samuelson, Ali Shourideh, Lars Stole, Roland Strausz, Bruno Strulovici, Olivier Tercieux, Kaizheng Wang, Assaf Zeevi, and participants at numerous seminars and conferences.  
We acknowledge research assistance from Mingeon Kim.  Yeon-Koo Che is supported by the Ministry of Education of the Republic of Korea and the National Research Foundation of Korea (NRF-2024S1A5A2A0303850912).}}

\author{ Yeon-Koo Che \\ Columbia University
\and Andrew B. Choi \\ University of Michigan
}
\date{February 2026 
}
\maketitle

\begin{abstract}

Motivated by applications such as cloud computing, gig platforms, and blockchain auctions, we study optimal selling mechanisms for dynamic markets with stochastic supply and demand. In our model, buyers with private valuations and homogeneous goods arrive stochastically and can be held in queues at a cost. The optimal mechanism pairs {\it allocative efficiency} with {\it dynamic admission control}: goods are assigned to the highest-value buyer, while entry is restricted by value thresholds that strictly increase with the queue length and decrease with available inventory. This policy smooths competitive pressure across time and is implemented in dominant strategies via auctions with dynamic reserve prices.

\vspace{2 mm}
\end{abstract}

\section{Introduction}

Classic auction models, pioneered by \cite{myerson1981} and \cite{milgrom1982theory}, focus on static settings with fixed supply and demand. In contrast, modern platforms match demand and supply arriving asynchronously and stochastically. This feature is ubiquitous: in cloud computing (e.g., AWS), capacity fluctuates with job completion times; in gig platforms, labor supply varies with demand; and blockchains must continuously sequence randomly arriving transactions. These dynamics require a departure from static design to one that actively regulates market thickness over time.

We capture this feature using a continuous-time model in which homogeneous goods and heterogeneous buyers arrive via independent Poisson processes with rates $\mu$ and $\lambda$, respectively.
Buyers draw private valuations $v \in [0,1]$ from a distribution $F$. The seller can maintain a queue of buyers at waiting cost $c > 0$ per unit time, and an inventory of goods at cost $d>0$ per unit time.

The model encompasses distinct real-world scenarios. First, {\it service centers} (e.g. cloud computing or repair facilities) correspond to the perishable case ($d=\infty$). Here, the arrival of a good is interpreted as the completion of service for a customer who is already receiving service; ``service completion'' is wasted if no customers are present upon its arrival. Second, {\it dynamic matching platforms} fit the general framework ($c, d<\infty$) in which both goods and buyers can be queued. This captures applications ranging from ride-sharing and freight ports to organ donation networks, where the platform uses queues to resolve supply-demand mismatches.\footnote{We focus on policies where allocation is not delayed, implying the buyer queue and item inventory are never simultaneously non-empty. We show that this restriction is without loss in large markets.}

The seller commits to a {\bf Markovian policy}  governing the entry and exit of buyers, inventory of goods, and their allocation. These decisions are contingent on the current {\bf state}---specifically, the {\it number of buyers} in the queue, their {\it reported valuations}, and the {\it inventory level}. Each policy induces a Markov chain over this state space. Our goal is to find the policy that maximizes the seller's expected revenue under the stationary distribution. As we argue in \Cref{sec:appendix-markov-wlog}, this focus on Markovian policies is without loss of optimality.\footnote{A rough intuition for this is that the arrival processes of buyers and items are Markovian, so any information other than the state is payoff irrelevant.}

 The optimal policy works as follows. When a good arrives, it is allocated to the highest-value buyer in the queue; if the queue is empty, the good is stored up to a limit $L^*$, beyond which it is discarded. When a buyer arrives, the allocation decision hinges on the inventory level $\ell$. If inventory is available ($\ell > 0$), the buyer receives a good if his valuation exceeds a threshold $\hat{v}_{-\ell}$. If the inventory is empty ($\ell=0$) and  $k-1$ buyers are waiting, the queue grows to $k$ if {\it all} $k$ buyers---including the new arrival---have valuations above a threshold $\hat{v}_k$. Otherwise, the lowest-value buyer is removed, keeping the queue size at $k-1$.

The threshold $\hat{v}_j$ represents the seller's dynamic opportunity cost.
Crucially, $\hat{v}_j$ increases strictly in the queue length $j$, for both $j<0$ (inventory) and $j>0$ (buyer queue).  Admitting a marginal buyer functions as insurance against {\it buyer stockouts} and {\it depletion of competitive pressure}. When the queue is short, the insurance value is high: an additional buyer significantly lowers the risk of physical waste and of low competition, leading to a high rent accruing to the winner.  However, as the queue lengthens, the marginal value of this insurance diminishes, prompting the seller to become progressively more selective.  Effectively, the mechanism dynamically smooths out the demand-supply imbalance and competitive pressure across time.  When the market is thin, the mechanism ``saves'' the competitive pressure, and even ``borrows'' it from the future demand; when the market is thick, it ``discards'' it.

The optimal policy can be implemented in dominant strategies via dynamic auctions with reserve prices.
When inventory is available ($j=-\ell$, for some $\ell>0$), the seller simply charges a monopoly price $\hat{v}_{-\ell}$ against an incoming buyer.  When there is no inventory of goods, however, dynamic auctions are used to screen incoming buyers into a queue.  Specifically,    a {\bf survival auction} is run where a rising clock price filters out low-value agents, establishing a personalized reserve price for those who remain.   When a good arrives with buyers waiting in the queue, an {\bf assignment auction} is held in which the highest bidder wins the good and is charged a price that is calibrated to eliminate the ex post regret of paying more than is necessary to eventually win; we formally define the pricing rule in \Cref{sec:DSIC-EPIC}.

This pricing rule smooths the level of competition across time. High congestion in the past raises the personalized reserve price, thereby funneling more competitive pressure toward the future. Conversely, high future congestion raises the cutoff price for the current assignment, allowing the seller to borrow competitive pressure from the future to discipline the current winner. Consequently, the winner's payment depends on whether the assignment is followed by a period of high demand or surplus supply, potentially requiring delayed payment.  However, we show that such delays can be avoided if the solution concept is weakened to periodic ex-post implementation \citep{bergemann2010dynamic}, where truth-telling remains optimal regardless of the state, provided future agents report truthfully.

We extend the analysis in two directions. First, we characterize the large-market properties of the optimal policy. As the market ``thickens'' ($\lambda, \mu \to \infty$) or as frictions vanish ($c, d \to 0$), the optimal policy converges to a familiar static benchmark: a {\bf multiunit uniform-price auction with a monopoly reserve}. This bridges our dynamic framework with classic auction theory.

Second, we study welfare maximization. The optimal policy retains the threshold structure but exhibits a single-crossing property. When the queue is short, the (welfare-maximizing) planner lowers thresholds to prevent physical waste due to buyer stockouts. Conversely and perhaps surprisingly, when the queue is long, the planner sets higher thresholds than the seller. This occurs because the seller values the competitive pressure of a marginal buyer, which reduces the winner's information rent. In contrast, the planner treats rent as a neutral transfer and focuses solely on balancing match quality against waiting costs.

The current paper makes two broad contributions. First, we develop a tractable workhorse model of dynamic mechanism design that captures the canonical features of platform markets: the {\it asynchronous} and {\it stochastic} arrival of demand and supply. This formulation expands the scope of the field: beyond the classic {\bf allocation} problem, the designer must now solve the novel problem of {\bf regulating queues} to provide future competition. Second, we advance the steady-state approach to dynamic mechanism design. While this approach has precedents, our environment is distinguished by the richness of the state space---the evolving, infinite-dimensional profile of private valuations. We develop a method to collapse this complex dynamic problem into a tractable static linear program by characterizing constraints on marginal stationary distributions. This technique resolves the ``curse of dimensionality'' in steady-state design and may prove useful for applications well beyond the current setting.

The rest of the paper is organized as follows. \Cref{sec:setup} describes the model. \Cref{sec:relaxed-program} formulates and solves the relaxed program, and \Cref{sec:optimal-mechanism} characterizes the optimal policy. \Cref{sec:storability} extends the analysis to storable goods ($d<\infty$). \Cref{sec:DSIC-EPIC} provides the dominant strategy implementation. \Cref{sec:large-market} and \Cref{sec:welfare} analyze large markets and welfare maximization, respectively. We defer a detailed review of the related literature to \Cref{sec:lit-review} to better contextualize our contribution within the technical landscape. \Cref{sec:conclude} concludes. All proofs are relegated to the Appendices.

\section{Setup}\label{sec:setup}

\paragraph*{Primitives.} 
  Time $t\in [0,\infty)$ is continuous. A single seller allocates homogeneous goods arriving at Poisson rate $\mu$ to buyers arriving at Poisson rate $\lambda$. Upon arrival, each buyer privately draws a unit-demand valuation $v \in V \coloneqq [0,1]$ from a distribution $F$ with density $f>0$. The arrivals of goods and buyers, and the buyers' valuations, are all independent. We assume $F$ is regular, meaning the virtual value $J(v) := v - (1-F(v))/f(v)$ is strictly increasing in $v$, and $J(0)<0$. We also assume that $f$ is absolutely continuous.

  When a buyer arrives but is not immediately allocated a good, the seller may ask him to join a queue. If the buyer declines, he leaves with zero payoff. If he joins, he stays until he is allocated a good, leaves voluntarily, or is removed by the seller. Buyers cannot come back after leaving. If a buyer with value $v$ spends time $w$ in the queue, makes a payment $\tau$, and receives the good, his payoff is $v-\tau-cw$, where $c>0$ is the per-unit time waiting cost.\footnote{By assuming that a buyer stops incurring the waiting cost once he leaves the queue, even if he has not received the good, we are considering environments in which buyers have outside options. Customers often have alternative, perhaps less valuable, options available for substitution. For example, a buyer of cloud computing can look for other services, a customer of a gig platform can switch to a different one, ships queueing to unload their freights at a port can head out to a different nearby port, and a kidney patient may opt for a kidney transplant from a blood-type incompatible donor instead of waiting in a queue for an exchange with a cross-compatible patient-donor pair.}  If he spends time $w$ and pays $\tau$ but receives nothing, his payoff is $-\tau-cw$. Buyers are aware of their own value and calendar time but do not observe the number of other buyers in the queue or their values. Goods held in inventory incur a flow cost $d>0$ to the seller.

This framework synthesizes the revenue-optimal auction design of \cite{myerson1981} with the rational queueing theory of \cite{naor1969regulation}.
By incorporating stochastic arrivals and waiting costs ($c>0$) into a setting with private valuations ($v \sim F$), the model captures the friction of temporal mismatch absent in the static auction model. Meanwhile, our ``service model'' ($d=\infty$) corresponds to Naor's M/M/1 queue, where arriving services are wasted if not immediately utilized.\footnote{The M/M/1 queue is a model of a queueing system where agent arrival times and service times follow independent exponential distributions, and at most one agent can be served at a time. When an agent in the M/M/1 model is being served at a given point in time, this can be equivalently described by stipulating that, if service completion were to arrive at that moment, it would be allocated to that agent.} More generally, our ``platform model'' ($d < \infty$) allows for inventory, capturing markets where both supply and demand can be stored to resolve mismatches.  In this general setting, we focus on policies with \textit{no allocation delay}, implying that the buyer queue and inventory are never simultaneously non-empty.
While this abstracts from the double-sided queueing seen in some applications, we show in \Cref{sec:large-market} that this restriction is without loss of optimality as the market grows large.

Our analysis proceeds in two steps. We first analyze the service scenario ($d=\infty$) before investigating the general platform model ($d<\infty$).
This sequencing serves a pedagogical purpose: the $d=\infty$ case isolates the mechanism design innovations---merging Myerson and Naor---without the added complexity of inventory management, allowing for a more transparent exposition.

\paragraph*{Seller's Markov Policy for $d=\infty$.}\label{subsec:seller's-dynamic-policy}

It is without loss for the seller to adopt a Markov policy that depends only on the current payoff-relevant states (see \Cref{sec:appendix-markov-wlog}). A \textbf{state} is a tuple $\mathbf{v}=(v_k)_{k\in\N} \in [0,1]^\N$,    where $v_k$ is the $k$-th highest valuation in the buyer queue. We define $v_k:=0$ if there are fewer than $k$ buyers in the queue.  In particular, if the queue is empty, then $v_1=0$.\footnote{Our definition of $\bv$ does not distinguish between a queue where some buyers have value 0, and another queue with no buyer present at that position. This distinction is irrelevant for our analysis.   \label{v_k:=0}}  A {\bf Markov policy} $\phi = (a, r, \tau)$ consists of the following measurable mappings:

\begin{itemize}
\item $a: \mathbf{V} \to \Delta(\mathbb{N} \cup \{0\})$, where for $k>0$, $a(\mathbf{v})(k)$ is the probability an incoming item is assigned to the the buyer with the $k$-th highest value, and $a(\mathbf{v})(0)$ is the probability it is discarded.\footnote{An item assigned to a non-existent (or zero-valued) buyer is effectively being discarded. Thus, without loss, we require that $a(\bv)(k)=0$ if $v_k=0$.}

\item    $r: \mathbf{V} \to \Delta(2^\mathbb{N})$, where $r(\mathbf{v})$ is the distribution over subsets of buyers immediately removed from the queue when the state reaches $\bv$.

\item  $\tau: \mathbf{V} \times V \to \mathbb{R}$, where $\tau(\bv,v)$ is the expected payment by a buyer who arrives at state $\bv$ and reports $v$, implemented by any stream of state-dependent payments over time.\footnote{At this point, for ease of exposition, we remain agnostic about the exact stream of payments that implements $\tau(\bv,v)$. This will be fully specified when we characterize the optimal policy.}

\end{itemize}
Note that we are not assuming work conservation---the seller may refuse to assign an (incoming) item, or withhold service, to buyers. 
The removal policy $r$ controls {\it both} entry and exit: not admitting a buyer's entry is encoded as the buyer being immediately removed upon entering the queue.

Since the buyers do not observe the queue state, in addition to choosing $(a,r, \tau)$, the seller may disclose information about the queue to existing or newly arriving buyers. Throughout, we focus on the optimal information policy, which is to disclose to each buyer no information beyond the action recommendation made to him at each instant.

The seller does not directly observe the buyers' values, and she cannot compel a buyer to join or stay in the queue against his wishes. Meanwhile, we allow the seller to prevent a buyer from joining the queue or force a buyer to leave the queue; this is without loss because the seller can implement it by promising never to allocate a good or charging a high enough payment to the buyer.
Following \cite{myerson1986multistage}, we restrict attention to incentive-compatible policies, under which each buyer is incentivized to truthfully reveal his value upon arrival and to follow the seller's recommendations to join or stay in the queue, as a Bayes-Nash equilibrium.

\paragraph*{Steady State Analysis.}  

A Markov policy induces a stochastic process $\{\mathbf{v}_t\}_{t\ge 0}$ over the state space.
We focus on a stationary distribution $\pi \in \Delta(\mathbf{V})$ of this process, which captures the system's long-run behavior.\footnote{Strictly speaking, this restricts the set of policies to those admitting stationary distributions. This restriction involves little loss. As we argue in \Cref{sec:appendix-markov-wlog}, the process ``regenerates'' each time the queue empties. Any policy that maximizes revenue must ensure this regeneration occurs in finite expected time (positive recurrence) to avoid infinite waiting costs. This property guarantees the existence of a unique stationary distribution \citep{asmussen2003applied, THORISSON1992237}.}
Consequently, we evaluate each policy based on the expected revenue generated under its associated stationary distribution $\pi$, which corresponds to the long-run average revenue.

Let $\Phi$ denote the set of all incentive-compatible policies admitting a stationary distribution on $\bV$. The \textbf{seller's problem} is
$$\sup_{\substack{\phi\in  \Phi \\ \pi \in \D (\bV)}}   \,   \int_{\bv} \int_0^1 \l \tau(\bv,v )   f(v)\ddd v\pi(\dd \bv), \leqno{[\mathcal P_0]}$$
where $\pi$ is a stationary distribution induced by $\phi$.

\begin{rem}(Sufficiency of Stationary Markov Policies) One might wonder if the seller could improve revenue by adopting a non-stationary policy or one that depends on the history of arrivals beyond the current state. In \cref{sec:appendix-markov-wlog}, we show that restricting attention to stationary Markov policies is without loss of optimality. First, the system ``regenerates'' whenever the queue empties, so the seller faces an identical  problem at the start of each busy cycle.  Therefore, it is without loss for a policy to be stationary across cycles.  Second, the arrival processes are memoryless, so  the ``state'' $\bv$ is a sufficient statistic for the future evolution of the system. Hence, any history not encoded in $\bv$ is payoff-irrelevant. For example, we could allow the seller to assign items based on buyers' arrival orders (e.g., First-Come, First-Served), but the seller would not benefit from conditioning on such non-Markovian history.
\end{rem}

\section{Feasibility Constraints and Relaxed Program}\label{sec:relaxed-program}

 Solving $[\mathcal{P}_0]$ directly is intractable. Instead, we adopt a relaxed optimization approach. 
Rather than optimizing over complex dynamic policies, we optimize over the induced variables---allocations, payments, and stationary distributions---subject only to certain necessary conditions.
These conditions, which effectively bridge standard mechanism design with queueing dynamics, are described in \Cref{subsec:drm}, \Cref{subsec:rf}, and \Cref{subsec:balance}.

\subsection{Direct Revelation Mechanism} \label{subsec:drm}

For each incoming buyer, a policy $\phi \in \Phi$ and its associated stationary distribution induce a \textbf{Direct Revelation Mechanism (DRM)}, defined as   measurable functions $(X, T, W): V \to [0,1] \times \mathbb{R} \times \mathbb{R}_+$, where
\begin{itemize}
    \item $X(v)$ is the \textit{interim allocation probability} that the buyer obtains a good,
    \item $T(v)$ is the \textit{interim expected payment} that the buyer makes to the seller,
    \item $W(v)$ is the \textit{interim expected waiting time},
\end{itemize}
if the buyer reports $v$.\footnote{The functions $X$, $T$, and $W$ are obtained as expectations over the state $\bv$ at the time of arrival (which follows the distribution $\pi$), as well as the future realizations of the stochastic process $\{\bv_t\}_{t\geq 0}$ induced by $\pi$ and $\phi$.}
Rather than optimizing over dynamic policies, we view the seller as choosing this interim mechanism $(X, T, W)$ directly to maximize revenue.
For such a choice to be valid, the mechanism must be implementable by some Markov policy that satisfies the buyers' incentive constraints.
This implies that the DRM itself must satisfy the buyers' \textbf{participation constraint}:
\begin{align}
   U(v):= vX(v) - T(v) - cW(v) \geq 0 \qquad \forall v\in V. \tag{$IR$} \label{IR}
\end{align}
and the buyers' \textbf{incentive compatibility constraint}:
\begin{align}
    U(v) \geq vX(v') - T(v') - cW(v') \qquad \forall v,v'\in V. \tag{$IC$} \label{IC}
\end{align}
Clearly, \eqref{IR} and \eqref{IC} must be satisfied in any equilibrium.\footnote{The RHS represents the payoff from mimicking a type-$v'$ buyer completely---reporting $v'$ and adhering to the exit recommendations associated with that report. This serves as a necessary condition for incentive compatibility. We ignore ``double deviations'' (e.g., reporting $v'$ but disobeying the exit recommendation) to formulate a relaxed problem. This relaxation is justified because the derived optimal policy will be shown to satisfy all incentive and obedience constraints.}
Initially, we only require that truthful reporting constitutes a Bayes-Nash equilibrium.
By adopting this weakest notion of implementation, we cast the widest net for identifying the optimal policy.
In \Cref{sec:DSIC-EPIC}, we will show that the optimal policy can be implemented in dominant strategies.

\subsection{Reduced-Form Feasibility Constraint} \label{subsec:rf}

A DRM  offered to a buyer must be compatible with the stochastic supply of goods and the DRMs offered to other buyers.
To represent this constraint, we consider two aggregate variables induced by the policy $\phi$ and its stationary distribution.

First, let $P_k(v):=\pi(\{\bv \in \bV \mid v_k \leq v\})$ denote the stationary CDF of the $k$-th highest value in the queue, under the chosen policy $\phi$.
This captures the distribution of the $k$-th highest valuation when an outside observer takes a ``snapshot'' of the queue at a random time \citep{wolff1982poisson}.
It follows that $p_k(v) := P_{k+1}(v) - P_k(v)$ represents the stationary probability that \textit{exactly} $k$ buyers in the queue have valuations strictly above $v$.\footnote{We let $p_0(v):=P_1(v)$. Since we treat a buyer with zero valuation as if he does not exist, $p_k(0)$ denotes the probability that there are exactly $k$ buyers in the steady state queue.}

Second, suppose the state is such that exactly $k$ buyers have valuations strictly above $v$.
Conditional on this event, and given the policy $\phi$, let $y_k(v)$ denote the probability that an arriving good is allocated to one of these $k$ buyers.

The \textbf{reduced-form feasibility constraint} links the interim allocation $X$ (offered to the buyer) to these aggregate variables (generated by the system):
\begin{align*}
   \lambda  \int^1_v X(s)f(s)ds
   \le \mu\sum_{k=1}^{\infty} y_k(v)p_{k}(v) \qquad \forall v\in V. \tag{$RF$} \label{RF}
\end{align*}
This condition requires that the total allocation offered to types above $v$ (LHS) cannot exceed the total rate at which goods are actually assigned to such types (RHS). While similar in spirit to the ``reduced-form auctions'' characterizations (see \cite{border1991implementation} and \cite{che2013generalized}), the constraint \eqref{RF} effectively embeds the temporal dimension of queueing into the standard Myersonian reduced-form approach.

\subsection{Balance Inequalities} \label{subsec:balance}

Recall that we restrict attention to policies that admit a stationary distribution.
Fully characterizing stationarity requires flow conservation to hold for \emph{every} measurable set in the infinite-dimensional state space.
Such a requirement is intractable.
Instead, we impose flow conservation only on a carefully curated collection of ``marginal'' sets.
Specifically, for any $k \in \mathbb{N}$ and $v \in [0,1]$, let $$S_k(v) \:= \{\mathbf{v} \mid v_k \le v\}.$$
This represents the event that the $k$-th highest valuation is weakly below $v$, or equivalently, that {\it at most $k-1$ buyers have values strictly above $v$}.
\Cref{fig:events-S_k(v)} depicts the projections of $S_1(v)$ and $S_2(v)$ into the first two coordinates of the state space.

\begin{figure}[hbt]
\centering
\begin{tikzpicture}
\begin{groupplot}[
    group style={
        group size=2 by 1, %
        horizontal sep=80pt, %
    },
    width=0.5\textwidth,
    xlabel={$v_1$},
    ylabel={$v_2$},
    xlabel style={below}, %
    ylabel style={rotate=90, above}, %
    xmin=0, xmax=1.15, %
    ymin=0, ymax=1.15, %
    axis lines=middle,
    axis equal,
    unit vector ratio*=1 1 1,
    xtick={0,1}, %
    ytick={0,1}, %
    grid=none, %
]

\nextgroupplot
\addplot[name path=A, domain=0:1, samples=100, color=black, thick] {x};
\addplot[name path=B, domain=0:1] {0}; %
\addplot [red!20] fill between[of=A and B, soft clip={domain=0:0.6}];
\node at (axis cs:0.4,0.15) {$S_1(v)$}; %
\node at (axis cs:0.99,1.02) {$v_1 = v_2$};
\pgfplotsset{xtick={0,1,0.6}, xticklabels={0,1,$v$}} %

\nextgroupplot
\addplot[name path=C, domain=0:1, samples=100, color=black, thick] {x};
\addplot[name path=D, color = white, domain=0.6:1] {0.6}; %
\addplot[name path=E, domain=0:1] {0};
\addplot [red!20] fill between[of=C and E, soft clip={domain=0:0.6}];
\addplot [red!20] fill between[of=D and E, soft clip={domain=0.6:1}];
\node at (axis cs:0.65,0.25) {$S_2(v)$};
\node at (axis cs:0.99,1.02) {$v_1 = v_2$};
\pgfplotsset{ytick={0,1,0.6}, yticklabels={0,1,$v$}} %
\addplot [color=black, dashed, thick] coordinates {(0, 0.6) (0.6, 0.6)};

\end{groupplot}
\end{tikzpicture}
\caption{$S_1(v)$ and $S_2(v)$}
\label{fig:events-S_k(v)}
\end{figure}

In steady state, the  instantaneous flow of probability mass leaving $S_k(v)$ must balance the  instantaneous flow entering it.
We capture a necessary condition for this balance via the following \textit{balance inequality}:
\begin{align*}
     \lambda[1-F(v)]   p_{k-1}(v)    \ge   \mu y_{k}(v) p_k(v).
\tag{${B_k}$}  \label{Bk}
\end{align*}

    The LHS represents the rate of arrival of a new buyer with a value strictly above $v$ when the queue currently has exactly $k-1$ such buyers. This is an \textbf{upper bound} on outflow from $S_k(v)$ because the seller might refuse entry to the new buyer or remove existing buyers.

The RHS represents the rate at which a buyer with value strictly above $v$ is served when the queue currently has exactly $k$ such buyers. This is a \textbf{lower bound} on the inflow into $S_k(v)$ because the seller might remove such buyers at will, which is another way to enter $S_k(v)$.

Crucially, this inequality formulation---rather than an equality---is what captures the seller's strategic discretion in her queue control.
The gap between the LHS (maximal potential outflow) and the RHS (minimal required inflow) represents the seller's freedom to reject arrivals or remove buyers. By optimizing subject to this inequality, the relaxed program will implicitly determine the optimal queue management policy.

Note that we enforce balance only on the sets $S_k(v)$ rather than all measurable sets, which makes \eqref{Bk} necessary but not sufficient for stationarity. This reduction from an infinite-dimensional problem to a lower-dimensional one is precisely what makes the optimization tractable.\footnote{These conditions apply only to lower-dimensional subsets of the state space. To characterize the full stationary distribution $\pi \in \Delta(\mathbf{V})$, one would need balance conditions for all measurable subsets of $\bV$, which has the cardinality of $2^{[0,1]^{\mathbb{N}}}$. By focusing on $S_k(v)$, which is indexed only by $(k,v)$, we reduce the cardinality to $\mathbb{N} \times [0,1]$. While not sufficient for defining $\pi$, these conditions will be appropriate for setting up the relaxed program $[\mathcal{P}]$  (to appear in \Cref{subsec:relaxed}), as they capture the binding constraints on the objective.}
These specific constraints will turn out to be the only ones that bind at the optimal solution.

\subsection{Relaxed Program} \label{subsec:relaxed}

We now assemble the constraints into a relaxed program.
Our goal is to maximize expected revenue over the choice variables $(p,y,X,T,W)$, subject to the necessary conditions derived above:
\begin{align*}
    &\max_{p,y,X,T,W} \l \int_0^1 T(v)f(v)dv\\
\mbox{   subject to } & \quad
\eqref{IR},\,\eqref{IC},\, \eqref{RF},\mbox{ and } \eqref{Bk} \,\, \forall k.
\end{align*}
One can see that this program is a relaxation of $[\mathcal{P}_0]$: the choice variables are induced by the policy $\phi$, and the constraints are necessary conditions for any valid policy.
By applying the standard envelope theorem to the incentive constraints (IR) and (IC), and appealing to Little's Law, we can bound the expected revenue as follows:
\begin{align}
 \l \int_0^1 T(v)f(v)dv
\le & \,  \l \int_0^1 \left[v X(v)-U(v)- c W(v)\right] f(v)dv\cr
= & \,  \l \int_0^1 \left[v X(v)-\int_0^v X(s)ds\right] f(v)dv - c \l \int_0^1 W(v) f(v)dv\cr
= &\,   \l \int_0^1 J(v) X(v) f(v)dv - c \sum_{k=1}^\infty k  p_k(0), \label{eq:envelope}
\end{align}
where $J(v)\:= v - \frac{1-F(v)}{f(v)}$ is the virtual value, and the last equality follows from Little's law.\footnote{Little's Law states that the average arrival rate ($\lambda$) multiplied by the average waiting time ($\int W(v)f(v)dv$) equals the average queue length.
At our steady state, the probability of there being exactly $k$ buyers is $p_k(0)$, so the average queue length is $\sum_{k=1}^\infty k p_k(0)$.}
Using this bound, we obtain the \textit{relaxed program}:
\begin{align*}
    [\mathcal P] \hspace{9em }&\max_{p,y,X} \, \l \int_0^1 J(v)X(v)f(v)dv - c \sum_{k=1}^\infty kp_k(0) \hspace{10em} \\
    &\quad \text{subject to }
  \, \eqref{RF}  \mbox{ and } \eqref{Bk}  \,\, \forall k.
\end{align*}
The structure of $[\mathcal{P}]$ is intuitive: the seller maximizes \textit{total virtual surplus} (the first term) net of \textit{total waiting costs} (the second term), constrained by the physical capacity of the system \eqref{RF} and the flow balance requirements \eqref{Bk}.  The optimal solution follows: 
 
\begin{thm} \label{thm:relaxed-program}   
There exists a triple $(P^*, y^*, X^*)$ that solves the relaxed program $[\mathcal P]$.
There exist thresholds $\hat v_1:=J^{-1}(c/\m)< \hat v_2 < ...< \hat v_{K^*}<1$, for some $K^*<\infty$, such that
\begin{itemize}
    \item the CDF $P^*_k(v)$ ($k=1, ..., K^*$) is supported on $\{0\}\cup [\hat v_k, 1]$ and is continuous and strictly increasing on $[\hat v_k, 1]$;
    \item the assignment rule is  {greedy}: $y^*_k(v) = 1$ for all $v\in[0,1]$ and $k \le K^*$;
    \item the interim allocation $X^*(v)$ binds \eqref{RF} for almost all $v \in [0,1]$.
\end{itemize} 
\end{thm}
\medskip

\Cref{sec:appendix-relaxed-program} provides a precise construction of $P_k^*(v)$ and $X^*(v)$.
The structure of the solution (see \Cref{fig:marginal-CDF} for the case of $K^*=2$) offers a preview of the optimal policy.
First, the optimal assignment policy $y^*$  always assigns a good to the highest-value buyer. In particular, this implies non-wastefulness: goods are never discarded if buyers are waiting.  Second, the support of $P_k^*$ implies that {\it the system accommodates a $k$-th buyer only if his valuation is above $\hat v_k$.} This suggests that the seller will enforce progressively stricter admission standards as the queue lengthens.

    \begin{figure}[htb]
        \centering
      \includegraphics[width=0.6\textwidth]{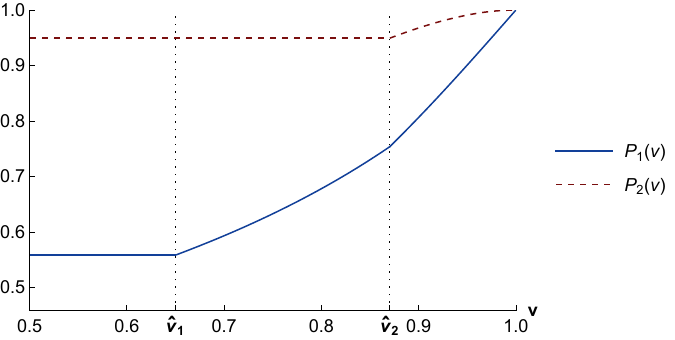}
        \caption{Stationary CDFs of order statistics ($F=U[0,1]$, $\mu=1$, $\l=2$, $c=0.3$)}   \label{fig:marginal-CDF}
\end{figure}

\section{Optimal Policy} \label{sec:optimal-mechanism}

We now construct a Markov policy $\phi^{*}\in\Phi$ that implements the optimal solution to the relaxed program $[\mathcal{P}]$. Since  $\phi^{*}$ achieves the revenue upper bound of $[\mathcal{P}]$ and satisfies all constraints of the original problem $[\mathcal{P}_0]$, it is optimal. 

\smallskip
\noindent $\bullet$ \textbf{Assignment:}  
    An arriving item is assigned to the buyer in the queue with the highest reported valuation. Formally, $a(\bv)(1)=1$ for all $\bv$. In the service context, service is always provided in descending order of reported values.
    
\smallskip
\noindent $\bullet$ \textbf{Buyer-queue management:}  At every instant, a buyer ranked $k$-th in the queue is retained if and only if his valuation is at least $\hat{v}_k$; otherwise, he is removed. Formally, $r(\bv)$ places probability 1 on removing the set of buyers $\{j, j+1, \dots\}$ where $j = \min \{k : v_k < \hat{v}_k\}$.\footnote{This removes all buyers with $v_k < \hat v_k$, since $v_k$ is non-increasing, while $\hat v_k$ is increasing.} On the equilibrium path, this implies a simple rule: when a new buyer arrives at a queue of size $k$, the buyer with the lowest valuation among the $k+1$ present (existing plus new) is removed if that lowest value is below $\hat{v}_{k+1}$; no other removals take place.

\smallskip
\noindent $\bullet$   \textbf{Payment:} Upon arrival, a buyer pays a lump-sum transfer $vX^*(v) - \int_0^v X^*(s)ds$. Note that this payment depends on the buyer's report $v$ but is independent of the queue state $\bv$. Additionally, while in the queue, each buyer is paid 
$c$ per unit time. 
\smallskip

The lump-sum transfer satisfies the standard envelope formula, ensuring that the policy is Bayesian incentive compatible and individually rational. Furthermore, the flow payment of $c$ fully reimburses the waiting cost, ensuring that buyers are willing to obey the recommendation to stay in the queue.\footnote{This flow reimbursement is feasible because the buyers' waiting cost is commonly known.}  We summarize these findings in the following theorem.

\begin{thm}  \label{thm:optimal-policy} 
The Markov policy $\phi^{*}$ induces a unique stationary distribution with marginals given by $P^{*}$, implements the allocation rules $y^{*}$ and $X^{*}$, and is incentive compatible. Consequently, $\phi^{*}$ is an optimal solution to the original problem $[\mathcal{P}_{0}]$.
\end{thm}

Several features of the optimal policy are noteworthy.
First, the assignment rule exhibits conditional efficiency: whenever a good arrives, it is allocated to the buyer in the queue with the highest valuation. This property follows directly from the regularity of the value distribution $F$. Since the virtual value function $J(v)$ is strictly increasing, maximizing expected revenue (virtual surplus) is achieved by prioritizing the agent with the highest valuation among those currently admitted.

 Second, the optimal queue length evolves as a classic birth-death process, increasing or decreasing by at most one unit at any instant. Although our framework permits the simultaneous removal of multiple buyers, such drastic corrections do not occur on the equilibrium path. Because the queue is always ``trimmed'' to satisfy the optimal thresholds, a new arrival triggers at most one removal---either the new buyer themselves or the lowest-value incumbent---ensuring the state transitions remain smooth.

 Third, the admission thresholds are {\it rank-deterministic}: the cutoff $\hat{v}_k$ depends solely on the buyer's rank $k$ and not on the exact valuations of the buyers ahead of them. This independence arises because, under the optimal assignment rule, a buyer's expected waiting time is determined exclusively by the {\it number} of higher-valuation buyers he must wait out, rather than by their specific values. Consequently, the seller's admission policy relies only on comparing the entrant's virtual value against the waiting costs associated with that specific rank.

 Finally, the thresholds $\hat{v}_k$ function as dynamic reserve prices. Like standard reserves, they exclude low-value buyers to limit information rents. Uniquely, however, these reserves are state-dependent and strictly increasing in queue length ($k$). Admitting a marginal buyer provides ``insurance'' against future stockouts and market thinness. When the queue is short, this insurance value is high, justifying a lower threshold. As the queue lengthens, the marginal benefit of this insurance diminishes, prompting the seller to become progressively more selective. Effectively, this policy smooths competitive pressure across time.

To verify that the optimal policy $\phi^*$ implements the interim allocation $X^*$, consider a buyer with valuation $v \in [\hat{v}_i, \hat{v}_{i+1})$ who joins the queue. Once admitted, the buyer's relative rank fluctuates according to the arrival of goods and more competitive buyers. Specifically, their rank improves by 1 whenever an item arrives (at rate $\mu$) and worsens by 1 whenever a new buyer with a higher valuation arrives (at rate $\lambda[1-F(v)]$). The buyer receives the good if their rank reaches 0 (the front of the queue) and is removed if their rank slips to $i+1$ (the exit threshold). Consequently, determining the interim allocation probability boils down to solving a classic {\it Gambler's Ruin} problem.\footnote{The Gambler's Ruin problem asks for the probability that a gambler starting with a given stake will reach a target wealth before going broke. In our context, the buyer's ``stake'' is their distance from the removal threshold, and the ``target'' is the service position. Gaining a chip corresponds to an item arrival (moving closer to service), while losing a chip corresponds to a higher-value buyer's arrival (moving closer to removal). The allocation probability is simply the likelihood that this random walk hits the service barrier before hitting the removal barrier.} The solution to this problem yields exactly $X^*(v)$.

The long-run behavior of the system is summarized by the joint stationary distribution of the value profile $\bv=(v_{k})$. \cref{fig:joint-CDF} visualizes this distribution for the first two order statistics $(v_{1},v_{2})$ in a setting where the maximum queue length is $K^{*}=2$. The distribution reveals the structure of the state space through three distinct regimes: a probability mass at the origin $(0,0)$ representing the empty queue; probability mass along the line segment defined by $\hat{v}_{1}\le v_{1}\le1$ and $v_{2}=0$, representing states with exactly one buyer; and the remaining probability spread over the triangular region $\{(v_{1},v_{2}) \mid v_{1}\ge v_{2}\ge\hat{v}_{2}\}$, representing states with two buyers.

\begin{figure}[htb]
    \centering
    \begin{subfigure}{.45\textwidth}
        \centering
        \includegraphics[width=\linewidth]{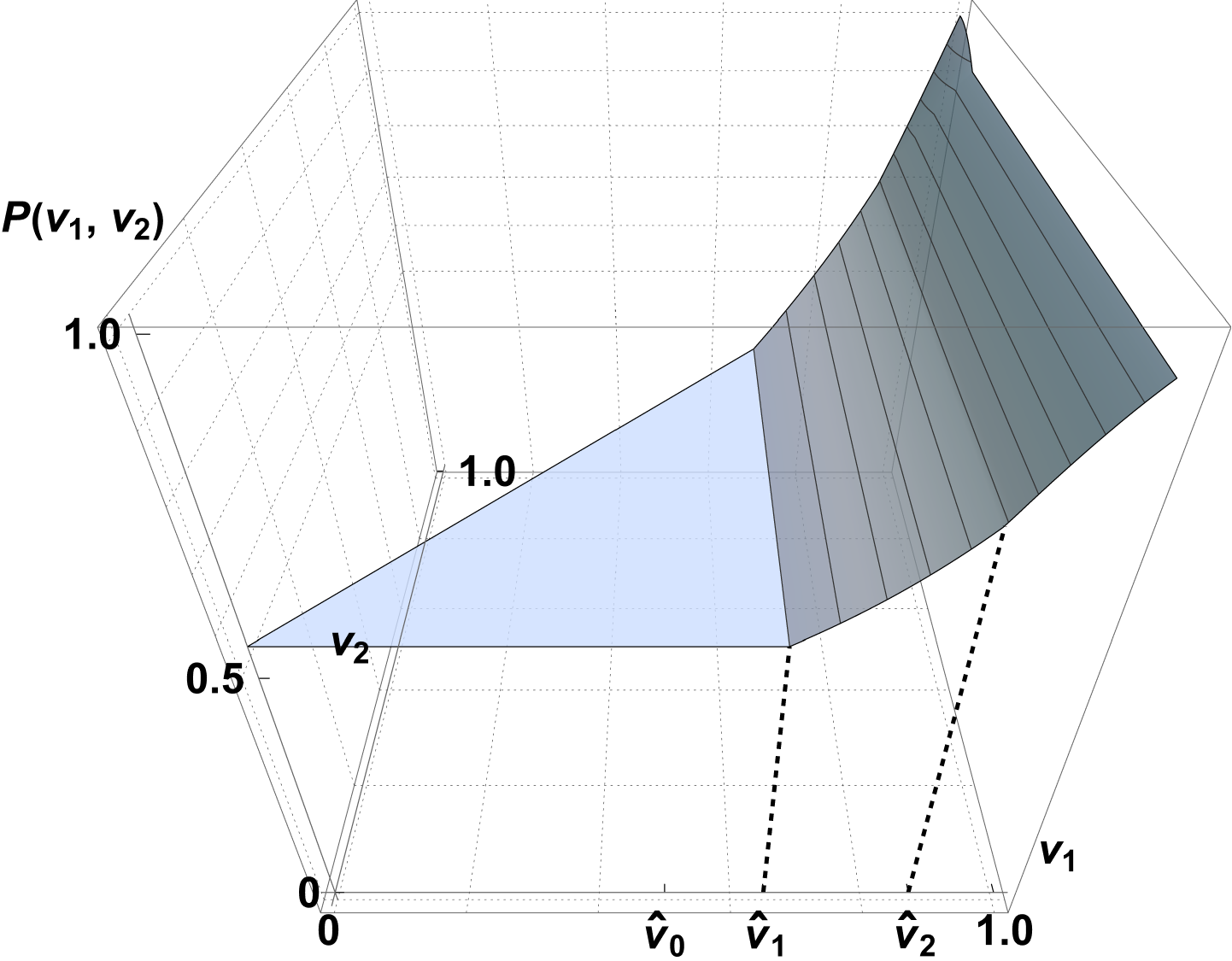} 
        \caption{CDF} \label{fig:cdf-3d}
    \end{subfigure}
    \hfill
    \begin{subfigure}{.4\textwidth}
        \centering
        \includegraphics[width=\linewidth]{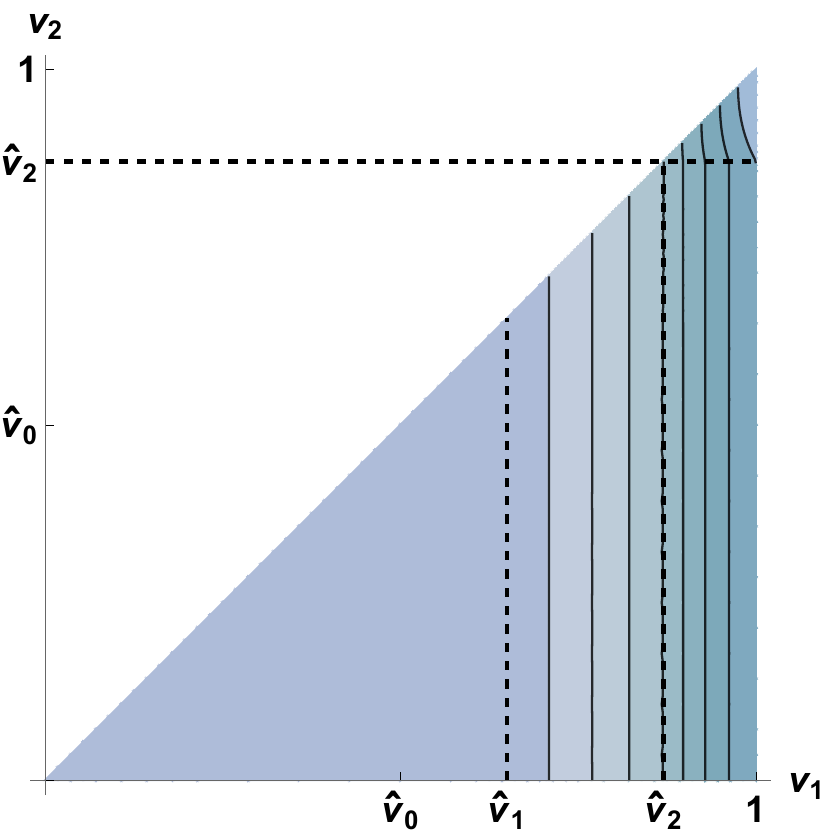} 
        \caption{Contour of the CDF} \label{fig:cdf-contour}
    \end{subfigure}
    \hspace{.1cm}
    \begin{subfigure}{.04\textwidth}
        \centering
        \includegraphics[width=\linewidth]{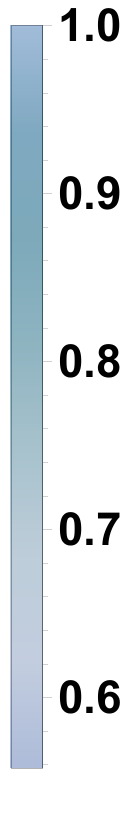}
        \caption*{} 
    \end{subfigure}
    
    \vspace{.5em}
    \raggedright
    \footnotesize{The parameters are $F=U[0,1]$, $\mu=1$, $\l=2$, and $c=0.3$. The maximum queue length under the optimal policy is $K^*=2$. Panel (a) shows the joint CDF of $v_1$ and $v_2$, which is given by $\Probability\{(v_1',v_2') \mid v_1' \leq v_1 \text{ and } v_2' \leq v_2\}$. Panel (b) shows the contour plot of the joint CDF.}
    \caption{Joint distribution of $v_1$ and $v_2$} \label{fig:joint-CDF}
\end{figure}

\paragraph*{Comparative Statics.}

The optimal thresholds respond intuitively to changes in the underlying parameters. 

\begin{prop}[Comparative Statics in $c,\l,\m$]\label{prop:thresholds-compstat}
     For $k=2, \dots, K^*$, the threshold $\hat{v}_k$ is strictly increasing in the waiting cost $c$, strictly increasing in the arrival rate $\lambda$, and strictly decreasing in the service rate $\mu$. The first threshold $\hat{v}_1$ is strictly increasing in $c$ and strictly decreasing in $\mu$, but is independent of $\lambda$.
\end{prop}

Intuitively, when waiting becomes more costly (higher $c$) or service slows down (lower $\mu$), the seller raises admission standards to avoid inefficient queuing. The response to arrival intensity $\lambda$ is more distinct. 
When buyers arrive more frequently (higher $\l$), a buyer who joins the queue is more likely to be replaced by a new buyer with a higher value. If a buyer joins an empty queue, he is served at a constant Poisson rate; therefore, the higher likelihood of replacement has no bearing on whether the seller should accept a buyer into an empty queue. Hence, $\hat v_1$ does not depend on $\l$.  However, if a buyer joins a nonempty queue, he is not served until all higher-valued buyers are served. When $\l$ is higher, the buyer is more likely to be replaced before he makes it to the top of the queue, so the seller is less willing to accept the buyer. This explains why $\hat v_k$ is increasing in $\l$ for $k>2$.

\section{Storable Goods: $d<\infty$}\label{sec:storability}

Thus far, we have focused on the perishable goods case ($d=\infty$). However, a defining feature of dynamic two-sided markets---such as gig platforms---is the ability for both sides to queue while waiting for a match. To capture this, we now relax the assumption of perishability and consider finite storage costs ($d<\infty$). This setting generalizes our previous analysis: as the storage cost $d$ grows sufficiently large, the optimal policy reduces to the service-only solution derived in \Cref{sec:optimal-mechanism}.

With finite storage costs, the seller can, in principle, maintain both a queue of waiting buyers and an inventory of stored goods simultaneously. To maintain tractability, however, we restrict attention to policies that clear one side of the market before accumulating the other.

\begin{assumption}[{\bf No allocation delay (NAD)}] When a buyer is queued, an arriving good is either assigned immediately or discarded. Conversely, when a good is in storage, an arriving buyer is either assigned immediately or turned away. \end{assumption}

This assumption implies that the buyer queue and the goods inventory are never simultaneously non-empty. Consequently, despite the infinite-dimensional nature of the value profile $\bv$, the system's ``length'' remains effectively one-dimensional: we can model the state space as a ``double-ended queue'' (or {\it dqueue}), where positive integers denote waiting buyers and negative integers denote stored goods.\footnote{That is, one can use $k\in \mathbb{N}$ to represent the existence of $k$ buyers in the queue, and use $k=-\ell$, for some $\ell\in \mathbb{N}$, to  denote $\ell$ goods in the inventory. While dqueues are standard in the Operations and Computer Science literatures, our formulation is significantly less restrictive. Unlike standard models that typically assume simple birth-death transitions, we allow for complex, value-dependent policies that may remove multiple buyers or discard multiple goods at any instant, provided the state does not ``jump'' across zero.} 
NAD is reminiscent of the ``Work-Conservation'' or ``Non-Wastefulness'' conditions often invoked in the queuing literature, though it is considerably weaker. In the service context ($d=\infty$), the assumption holds vacuously. While restricting simultaneous queues is not theoretically without loss,\footnote{The only theoretical reason to violate NAD---holding both a buyer and a good---is the option value of future abundance. A buyer currently deemed ineligible (value below the current threshold) might become eligible in the future if a surge of item arrivals lowers the opportunity cost of assignment. Violating NAD would allow the seller to retain this buyer for that contingency. However, for this to be optimal, the option value must outweigh the  holding costs ($c+d$) incurred during the delay. We have not identified numerical parameters where this benefit justifies the cost even when $c+d$ is low.  Indeed, as we argue below, NAD entails no loss of optimality as $c, d \to 0$.} the economic rationale for doing so with homogeneous goods is limited. Since goods are identical, delaying allocation to a waiting buyer yields no improvement in match quality. Indeed, we show in \Cref{sec:large-market} that imposing NAD is without loss of optimality as the market grows large ($\lambda, \mu \to \infty$) or as frictions vanish ($c, d \to 0$). This finding is in line with recent dynamic matching literature suggesting that immediate matching is approximately optimal in large markets \citep{akbarpour2017thickness,ashlagi2023matching}. 

The seller's dynamic policy must now also specify inventory management.\footnote{Formally, we augment the policy with two mappings: $\tilde{a}(v, \ell) \in [0,1]$, which determines the probability of serving an arriving buyer with value $v$ given inventory $\ell$; and $\tilde{r}(\ell) \in \Delta(\mathbb{N})$, which specifies the distribution over the number of discarded items when the inventory reaches level $\ell$.} To operationalize this, we expand the state space to $(\bv;\ell) \in \bV \times \Z_+$, where $\bv$ represents the value profile of buyers and $\ell$  the number of stored goods. Under the NAD assumption, the state effectively collapses to either $(\bv; 0)$ or $(\mathbf{0}; \ell)$. While the formulation of the Direct Revelation Mechanism remains unchanged, we must extend the reduced-form feasibility  \eqref{RF}  and balance conditions \eqref{Bk}'s to account for inventory. Let $q_{\ell}$ denote the stationary probability that the inventory contains exactly $\ell$ goods, and let $z_{\ell}(v)$ denote the probability that an incoming buyer with value $v$ is assigned a good conditional on the inventory level being $\ell$. The reduced-form feasibility constraint is then given by:
\begin{align*}
   \lambda  \int^1_v X(s)f(s)ds
   \le \mu\sum_{k=1}^{\infty} y_k(v)p_k(v) + \l \sum_{l=1}^\infty q_\ell \int_v^1 f(s)z_\ell(s)ds, \quad \forall v\in V. \tag{$\widetilde{RF}$} \label{RF-storage}
\end{align*}
The RHS now accounts for allocations from two sources: goods arriving when buyers are waiting (first term) and buyers arriving when goods are stored (second term).

The balance inequality for the buyer queue, \eqref{Bk}, remains formally unchanged, with the boundary term redefined as $p_{0}(v) := P_{1}(v) - \sum_{\ell=1}^{\infty}q_{\ell}$ to account for stored goods.\footnote{For the case of $d=\infty$, we had defined $p_0(v)\:=P_1(v)$.} For the inventory, let $S_{-\ell}$ denote the set of states with at least $\ell$ goods. A necessary condition for stationarity is: for each $\ell\ge 1$
\begin{align*}
     \l  q_\ell \int_0^1 f(s)z_\ell(s)\dd s \leq \mu q_{\ell-1}, \tag{${B_{-\ell}}$}  \label{B-l}
\end{align*}
where we let $q_0 \:= P_1(0) - \sum_{\ell=1}^\infty q_{\ell}$ represents the probability that the system is completely empty. The LHS represents the rate at which inventory decreases from $\ell$ to $\ell-1$ due to service (a buyer arriving and being assigned a good). This constitutes a lower bound on the outflow from $S_{-\ell}$ because the seller may also discard goods freely. The RHS represents the rate at which inventory increases from $\ell-1$ to $\ell$ due to supply arrival. This constitutes an upper bound on the inflow to $S_{-\ell}$ because the seller may choose to discard arriving goods rather than store them.\footnote{Whereas the balance inequality for $S_k(v)$ requires that outflow is no less than inflow, the balance inequality for $S_{-\ell}$ requires that outflow is no greater than inflow. 
}

We now state the relaxed program for the general case. Proceeding as in \Cref{subsec:relaxed}---applying the Envelope Theorem to substitute payments and Little's Law to rewrite waiting costs---we arrive at the following program:
\begin{align*}
    [\widetilde{\ms P}] \hspace{4em} \max_{p,q,y,z,X} \,\l  &\int_0^1  J(v)X(v)f(v)dv - c \sum_{k=1}^\infty k p_k(0)  - d \sum_{\ell=1}^\infty \ell q_\ell \hspace{10em}\\
   &  \text{subject to} \, (\widetilde{RF}),\, (B_k)\,\, \forall k, \, \mbox{and} \, (B_{-\ell})\,\, \forall \ell.
\end{align*}

The following is an analogue of \cref{thm:relaxed-program}.

\begin{customthm}{$\mathbf{1'}$}\label{thm:relaxed-program-storable}  Suppose $d<\infty$.  The relaxed program  $[\widetilde{\ms P}]$ admits an optimal solution $(P^*, q^*, y^*, z^*, X^*)$ characterized by strictly increasing thresholds, $\hat{v}_{-L^*} < \dots < \hat{v}_{K^*}$, such that 
\begin{itemize}
    \item the inventory distribution $q^*$ is supported on $\{0, \dots, L^*\}$, and the buyer queue distributions $P_k^*$ are supported on $\{0\} \cup [\hat{v}_k, 1]$ for $k \le K^*$;
    \item  the assignment rules are greedy: $y_k^*(v)=1$ when an item arrives while buyers are waiting, and $z_{\ell}^*(v) = \mathbb{I}\{v \ge \hat{v}_{-{\ell}}\}$ when a buyer arrives while inventory is available;
    \item the interim allocation $X^*(v)$ binds $(\widetilde{RF})$ for almost all $v \in [0,1]$. 
\end{itemize} 
\end{customthm}

    \Cref{sec:appendix-relaxed-program} provides the precise construction of these variables. Crucially, the support of $P^*_k$ implies that the system accommodates a $k$-th buyer only if his valuation exceeds $\hat{v}_k$. This property reveals that inventory and queue management are not distinct problems but rather two sides of the same coin. We can view the state space as a continuum of ``net demand''---the number of waiting buyers minus the number of stored goods. The strictly increasing thresholds $\hat{v}_{j}$ imply that a single ``law of demand'' holds. When net demand is low (high inventory, negative indices), the planner sets low price floors to clear stock and avoid holding costs. As net demand rises (inventory depletes and queues form), the planner smoothly transitions to raising admission standards to mitigate waiting costs. The optimal solution thus balances the system by strictly increasing the shadow price of allocation as the market tightens.

 As in \cref{sec:optimal-mechanism}, we can construct a Markovian policy $\phi^* \in \Phi$ which implements the solution to $[\widetilde{\ms P}]$.  The buyer-queue management, goods assignment, and payment rules remain unchanged, except for the threshold values. The policy  $\phi^*$ has a new component:  
 
\smallskip
\noindent $\bullet$   \textbf{Inventory management:} Suppose the system holds $\ell\ge 0$ units in inventory.  
    \begin{enumerate}
        \item [(i)] If a good arrives, it is stored if $\ell \le L^* - 1$ and is otherwise discarded;
        \item [(ii)] If a buyer with value $v$ arrives and $\ell \ge 1$, the buyer is allocated a good immediately if $v \ge \hat{v}_{-\ell}$ and is otherwise turned away.
    \end{enumerate}

The result follows:

\begin{customthm}{$\mathbf{2'}$}\label{thm:optimal-policy-storable}
    Suppose $d<\infty$.  The policy $\phi^*$ implements the optimal solution $(P^*, q^*, y^*, z^*, X^*)$ to the relaxed program $[\widetilde{\ms P}]$ and is incentive compatible. Therefore, $\phi^*$ is an optimal policy for the seller.
\end{customthm}

\section{Robust Implementations of Optimal Policy}\label{sec:DSIC-EPIC}

The optimal policy $\phi^*$ described in \Cref{thm:optimal-policy-storable} is Bayesian incentive-compatible. Here, we propose two ways to strengthen the robustness of $\phi^*$.

\subsection{Dominant Strategy Implementation} \label{subsec:DSIC}

We first present a \textbf{Cutoff Price Mechanism (CPM)} that implements $\phi^*$ in dominant strategies. For each buyer $i$, the mechanism maintains a {\bf personalized reserve price} $r_i$, initialized to $0$. The mechanism operates as follows:

\smallskip
\noindent {\bf (a) Inventory Sales:} If a buyer arrives when there is positive inventory $\ell \ge 1$, the seller posts a monopoly price of $\hat{v}_{-\ell}$. The buyer decides whether to purchase at that price or leave.

\smallskip
\noindent {\bf (b) Active versus Passive Status:} Buyers in the queue participate in multiple auctions over time. A buyer remains Active until participating in his first assignment auction. Upon submitting a bid $b$ in that auction, the buyer becomes Passive. From that moment on, the platform acts as a Proxy Agent for the buyer, using $b$ as the buyer's fixed bid for all future events:
    \begin{itemize}
        \item \textit{Future Survival Auctions:} If the clock price exceeds $b$, the Proxy Agent automatically drops out, and the buyer is removed.
        \item \textit{Future Assignment Auctions:} The Proxy Agent automatically submits a bid of $b$.\footnote{If the real bidder drops out before the Proxy Agent loses in a survival auction, and the Proxy Agent later wins a good, then the good is discarded.}
    \end{itemize}

\smallskip
\noindent {\bf (c)  Survival Auctions:} If a buyer arrives when the queue has $k-1 \geq 0$ buyers and no inventory, a survival auction is initiated with a common clock price rising from $\hat v_{k-1}$ (where we let $\hat v_0\:=J^{-1}(0)$).
    \begin{itemize}
        \item Each active buyer continuously decides whether to remain or drop out.
        \item For each passive buyer, his Proxy Agent drops out if the price exceeds his fixed bid $b$.
        \item \textit{Outcome:} When a buyer drops out, or the price reaches the admission threshold $\hat{v}_{k}$, the price clock stops, and the auction concludes. All remaining buyers stay.
        \item \textit{Reserve Price Update:} For every survivor $i$, the personalized reserve price is updated: $r_i \leftarrow \max\{r_i,\, \text{stopped clock price}\}$.
    \end{itemize}

\smallskip
\noindent {\bf (d)   Assignment Auctions \& Pricing:} 
Suppose a good arrives when the buyer queue is non-empty. If there is only one buyer, then he wins the good and pays his personalized reserve price.  Otherwise, an assignment auction is held.
    \begin{itemize}
        \item \textit{Bidding Constraint:} Each active buyer submits a sealed bid $b \ge r_i$. Each passive buyer bids his fixed bid $b$ (via his Proxy Agent).
        \item \textit{Allocation:} The good is allocated to the highest bidder, breaking ties randomly.
        \item \textit{Pricing:} The winning bidder pays his \textbf{Cutoff Price}. This price is defined as the lowest bid $\mathcal B \ge r_i$ such that, if the bidder were to adopt $\mathcal B$ as his fixed bid (and be represented by a Proxy Agent using $\mathcal B$ for all subsequent events), he would eventually win a good on the realized sample path of future arrivals. 
        \item \textit{Billing:} Payment is collected only after the Cutoff Price is determined (which may require observing future events).
    \end{itemize}

\smallskip
\noindent {\bf (e)   Waiting Costs:} Buyers are reimbursed for realized waiting costs, regardless of allocation outcomes.

\noindent {\bf (f) Information Policy:}  Buyers' identities and bids are private,  and buyers do not observe the queue, including its length.   

\medskip

To illustrate how the cutoff price is determined and how the optimal policy manages competitive pressure over time, we consider three example scenarios surrounding the auction between buyer $A$ (with value $v_A=6$) and buyer $B$ ($v_B=3$), as depicted in \Cref{figure:DSIC-example}. Assume the optimal thresholds are $\hat{v}_1=1$, $\hat{v}_2=2$, and  $\hat{v}_3=4$. In all three scenarios, $A$ wins under truthful bidding, but the cutoff price varies based on past or future events.

\begin{figure}[htb]
\centering
\begin{tikzpicture}[>=stealth, thick, xscale=1.1]
    \draw[->] (-1,3) -- (10,3) node[right] {\small \quad Scenario 1};
    \foreach \x/\label in {-0.5/$v_A=6$, 4/$v_B=3$, 6/item, 8/\gray{item}} {
        \draw (\x, 3.1) -- (\x, 2.9) node[below] {\small \label};
    }

    \draw[->] (-1,1.5) -- (10,1.5) node[right] {\small \quad Scenario 2};
    \foreach \x/\label in {-0.5/$v_A=6$, 4/$v_B=3$, 6/item, 7.5/\gray{$v_C=5$}, 9/\gray{$v_D=5$}} {
        \draw (\x, 1.6) -- (\x, 1.4) node[below] {\small \label};
    }

    \draw[->] (-1,0) -- (10,0) node[right] {\small \quad Scenario 3};
    
    \foreach \x/\label in {-0.5/$v_A=6$, 4/$v_B=3$, 6/item, 8/\gray{item}} {
        \draw (\x, 0.1) -- (\x, -0.1) node[below] {\small \label};
    }
    
    \draw[thin] (-.5, -1.2) -- (4, -1.2);
    
    \foreach \x/\label in {-0.2/$v_{\alpha}=8$, 1/$v_{\beta}=7$, 2.2/item, 3.4/item} {
        \draw[] (\x, -1.15) -- (\x, -1.25) node[below, black] {\small \gray{\label}};
    }
    
    \draw[dotted, gray] (1, -0.1) -- (-.5, -1.15);
    \draw[dotted, gray] (2.5, -0.1) -- (4, -1.15);

\end{tikzpicture}
\caption{Three example scenarios}
\label{figure:DSIC-example}
\end{figure}

\smallskip 
\noindent{\bf $\bullet$  Scenario 1 (Baseline):} Buyers $A$ and $B$ arrive at an empty system. The queue grows to 2, raising the survival clock to  $\hat{v}_2=2$. Both buyers survive and face a personalized reserve price of $r_A=r_B=2$.  Upon arrival of an item, an assignment auction is held, and $A$ wins.  To determine the cutoff price, the mechanism runs a counterfactual: ``What is the lowest bid $A$ could have made and still eventually won?''
Since another item arrives shortly after, $A$ would have won eventually provided he bid at least $2.$ Thus, the cutoff price is 2.

\smallskip 
\noindent{\bf $\bullet$  Scenario 2 (Borrowing from Future):}  This scenario mirrors Scenario 1, except for the post-auction arrivals. Instead of a second item,
two high-value buyers, $C$ and $D$ ($v=5$)  arrive.  This alters the counterfactual:     If $A$ had bid $b < 3$ (and lost to $B$), he would have remained in the queue only to be removed upon the arrival of $C$ and $D$ (due to the higher threshold $\hat v_3=4>3>b$).  Thus, to win at all, A must beat $B$ immediately, so the cutoff price becomes 3. Here, the mechanism effectively uses the threat of future arrivals ($C$ and $D$) to discipline the current winner.

\smallskip 
\noindent{\bf $\bullet$ Scenario 3 (Saving from Past):} This scenario mirrors Scenario 1, except for pre-auction events. After $A$ arrives, high-value buyers $\alpha$ ($v=8$) and $\beta$ ($v=7$) enter. The queue grows to 3, increasing the survival threshold to $\hat{v}_3=4$.  Buyer $A$ survives, updating his personalized reserve price  to $r_A=4$. Two items arrive immediately, serving $\alpha$ and $\beta$. Buyer $A$ remains and subsequently competes against $B$ ($v_B=3$).  As in Scenario 1, ample supply follows.  Yet, A's cutoff price is $4$. Although the current competition is weak, $A$ pays 4 because he previously survived through a period of high congestion that required a valuation of 4. The mechanism leverages this past competition to set a higher price floor.  \qed

\smallskip

 Unlike in the standard Vickrey auction with a reserve, our cutoff price varies dynamically. When competition is low, the cutoff price is below the second-highest bid (Scenario 1); when competition is high, it is above the second-highest bid (Scenario 3).
 Still, it serves the same purpose: to eliminate the ``hindsight'' regret of paying more than is necessary to win. As the examples illustrate, establishing this regret-free price requires assessing the market outcome well after the auction concludes.\footnote{The maximum delay can be bounded by the length of the cycle between two emptying episodes, which is finite in expectation.} As will be shown later, delayed billing can be avoided, but only at the expense of dominance, since incentives will then rely on the buyer's beliefs about the future arrival process and equilibrium behavior.

\begin{thm} \label{thm:DSIC}
The CPM is strategyproof.  That is, it is a (weakly) dominant strategy to (i) accept the good in (a) if and only if one's value is above the monopoly price, (ii) drop out of the survival auction in (c) if and only if the clock price has reached one's value, and (iii) bid one's valuation in the assignment auction in (d).
\end{thm}  

The dominance property provides a kind of robustness that is often absent in steady-state mechanism design.  Because truth-telling is optimal state-by-state, a buyer's incentives depend neither on his beliefs, others' strategies, nor on the system being in the stationary distribution.
Crucially, CPM achieves this robustness through two restrictions: opacity and passivity.
By withholding information about the queue and auction and forcing buyers to be passive after the initial bid, the mechanism ensures that a buyer's deviation cannot affect his opponents' future actions in a way that lowers the price for him.  See \cref{sec:appendix:reason-for-passive-buyers} for examples that show dominance fails without these restrictions.\footnote{If CPM disclosed information about buyer identities, past bids, or queue length, it would still satisfy the requirements of \textit{as-if dominance} \citep{nagel2025if} (except that they assume a finite game, whereas our game is infinite horizon). A mechanism is as-if dominant if agents possess an optimal strategy provided they neglect others' ability to condition on fine-grained history (a ``coarsely optimal'' strategy), and if such strategies constitute an ex post equilibrium.}  This informational firewall effectively isolates each buyer's decision problem, ensuring that incentive compatibility and obedience are preserved even if the system is initialized from an arbitrary state or if buyers hold misspecified beliefs about the underlying stochastic process.

This robustness extends to the seller's revenue. Starting from any finite state (e.g., $k=\ell=0$), the stochastic process $\{\bv_t\}_{t\geq 0}$ induced by the optimal policy converges to the unique stationary distribution, and the long-run time-average revenue converges almost surely to the optimal value of the relaxed program $[\widetilde{\ms P}]$ \citep{asmussen2003applied, THORISSON1992237}.  This provides a strong performance guarantee: the mechanism is not only incentive-compatible along any sample path but also revenue-optimal in the long run:

\begin{cor}  Under the cutoff-price auction, the long-run time-average of the seller's revenue converges almost surely to the value of $[\widetilde{\mathcal P}]$ as $T\to \infty$.
\end{cor}

\subsection{Periodic Ex Post Implementation}\label{subsec:periodic-EPIC}

One feature of our dominant-strategy implementation is the possibility of payment delays. While delayed payments are not uncommon,\footnote{Customers are often billed for later payments after the service is rendered.  This may be enforced via legal liability or by requiring buyers to post bonds that are returned after payment.} one may wonder if such a delay can be avoided.  Our original version of the direct mechanism, which specifies a lump-sum transfer upon each buyer's arrival, involves no delay and is Bayesian incentive compatible. In fact, delayed payment can be avoided while retaining robustness.

 Consider the \textbf{Expected Cutoff Price Mechanism (ECPM)}, which is the same as the Cutoff Price Mechanism except for the pricing rule: upon winning, a buyer is charged the \textit{expected} cutoff price conditional on the current state, where the expectation is taken over the stochastic arrival of future goods and buyers, assuming all future agents report truthfully.

This mechanism is no longer dominant strategy incentive compatible, since the payment depends on the assumption that future agents play equilibrium strategies. However, it satisfies a robustness property stronger than standard Bayesian incentive compatibility, known as \textbf{periodic ex post incentive compatibility} \citep{bergemann2010dynamic}.

\begin{defn} A (direct) mechanism is {\bf periodic ex post incentive compatible} if truth-telling is a best response for each agent, regardless of the history and the current state of the queue, given that all other agents report truthfully in the future.\footnote{To be precise, periodic ex-post incentive compatibility, as defined by \cite{bergemann2010dynamic}, does not require dominance for the current period; in this sense, our definition is stronger. }    
\end{defn}
The qualifier ``periodic'' requires that the buyer form a correct belief about the future states when making his bids.

\begin{cor}
The ECPM is periodic ex post incentive compatible.
\end{cor}

\begin{proof}  Since the cutoff price is dominant strategy incentive compatible for each state and for all possible behavior by future buyers, its expectation over future states preserves incentives for truthful reporting so long as all future buyers report truthfully and each buyer forms accurate beliefs about the future states.  
\end{proof}

As observed by \cite{ausubel2004}, an important advantage of ex post implementation is that it allows for \textbf{full transparency}.
Recall that dominance requires opacity and passivity to ensure that a deviation today does not alter the strategies of future agents in a way that would lower the price.  By contrast, periodic ex post implementation assumes that future agents play truthfully, regardless of the history. Because this assumption effectively ``freezes" the predicted behavior of future agents, disclosing the full history of buyers' behavior as well as allowing them to remain active throughout, cannot perturb the buyer's beliefs about future play. Consequently, this (admittedly weaker) notion is compatible with full transparency.

\section{Large Market Properties}\label{sec:large-market}

Next, we investigate the large-market properties of our optimal policy. Specifically, we ``thicken'' the market by letting the arrival rates $\lambda, \mu \to \infty$ while holding the balance parameter $\rho = \lambda/\mu$ constant. We show that the optimal policy converges to a familiar static benchmark: a multi-unit uniform-price auction with a fixed monopoly reserve.  Notably, the same result holds if either $c$ or $d$ tends to zero.  Since the static benchmark provides a revenue upper bound (as will be seen shortly), these results mean that NAD is with little loss in a large market.  

\paragraph*{Oracle Benchmark.}  Consider a hypothetical setting where $n$ buyers, each with a value drawn independently from $F$, compete for $m$ units of the good. Under the standard regularity condition, the $(m+1)$-th price uniform auction with a reserve price set at $\hat{v}_0 := J^{-1}(0)$ maximizes revenue. The auction admits a dominant strategy of truthful bidding, so the winners pay the $(m+1)$-th highest valuation (set equal to zero if $m \ge n$) or the monopoly price $\hat{v}_0$, whichever is higher. Let $R^*(m,n)$ denote the auction's expected revenue.

We argue that this static benchmark establishes an upper bound on the expected revenue attainable in our dynamic environment. To formalize this comparison, consider the underlying stochastic process as generated in two steps: first, realize a sample path of arrivals governed by the rates $\lambda$ and $\mu$; second, draw a valuation from $F$ for each arriving buyer. Let $(m, n)$ index a sample path where $m$ items and $n$ buyers arrive during a unit interval, and let $R_{\phi}(m, n)$ denote the expected revenue generated by a policy $\phi$ on such a path. Since the static oracle $R^*(m, n)$ optimizes the allocation for this fixed set of agents without timing constraints, it necessarily outperforms any dynamic policy subject to sequential arrival frictions.

\begin{lem}\label{lemma:large-market-1}
    For any sample path $(m, n)$ and any policy $\phi \in \Phi$, the expected revenue satisfies $R_{\phi}(m,n)\le R^*(m,n)$.
\end{lem}

Now,   let $n,m \to \infty$ such that $n/m\to \rho:=\l/\m$. Then, the selling price converges almost surely to
$\tilde v_0:=\max\{\tilde v, \hat  v_0\},$
where $ \tilde v:=\inf\{v: \rho [1-F(v)]\le 1 \}.$
Consequently, the seller's per-unit revenue $R^*(m,n)/m$ converges almost surely to
$$R^*:=\rho[1-F(\tilde v_0)]\tilde v_0.$$

\begin{lem} \label{lemma:large-market-2}  The expected per-unit revenue achieved by any policy $\phi \in \Phi$ is no greater than $R^*$ in the limit as $\lambda, \mu \to \infty$ with $\rho = \lambda/\mu$ held constant. 
\end{lem}

We now present our main result for this section.

\begin{thm}\label{thm:large-market}
The expected per-unit revenue under the optimal policy $\phi^*$ converges to the oracle benchmark $R^*$ if (a)
 $\l,\m\to \infty$ with $\rho=\l/\m$ held constant, or (b) $c \to 0$, or (c) $d \to 0$. 
\end{thm}

To understand why these conditions are related, observe that multiplying the parameters $\lambda, \mu, c,$ and $d$ by a constant $k > 0$ is equivalent to rescaling time, which does not affect the optimal policy or per-unit revenue. 
Therefore, the vanishing cost limit $c \to 0$ (case (b)) is isomorphic to a large-market limit $\l,\m,d \to \infty$ with $\l/\m$ and $\l/d$ held constant, and $d \to 0$ (c) is equivalent to $\l,\m,c \to \infty$ with $\l/\m$ and $\l/c$ held constant. Hence, the general large-market convergence (a) follows directly from establishing either (b) or (c).

To prove (b),  consider a simple  policy which assigns an incoming good to a waiting buyer if the buyer queue is non-empty and discards the good otherwise, and queues a buyer if and only if his value is above $\tilde v_0 + \d$ for an arbitrarily small $\d>0$ (see \cref{fig:large-market}). This ensures that the service rate $\mu$ is larger than the ``effective arrival rate'' $\l[1-F(\tilde v_0 + \d)]$.  This ``slack'' guarantees that every admitted buyer is eventually served, while the stable queue size ensures that average waiting times remain bounded. Consequently, as the waiting cost parameter $c \to 0$, the efficiency loss from queuing vanishes, and the per-unit revenue converges to the oracle benchmark $R^*$.   Since this simple policy achieves the bound, the optimal policy must necessarily do so as well.

\begin{figure}[H]
\centering
\begin{tikzpicture}
\begin{groupplot}[
    group style={
        group size=1 by 1, %
        horizontal sep=40pt, %
    },
    height=19em,
    xlabel={},
    ylabel={$v$},
    xlabel style={below}, %
    ylabel style={rotate=90, above}, %
    xmin=0, xmax=1.1, %
    ymin=0, ymax=1.1, %
    axis lines=middle,
    axis equal,
    unit vector ratio*=1 1 1,
    xtick={0,1}, %
    ytick={0,1}, %
    grid=none, %
]

\nextgroupplot
\addplot[name path=A, domain=0:1, samples=100, color=black, thick] {1-x};
\addplot [name path = B, color=black, dashed] coordinates {(0, 0.5) (0.5, 0.5)}; %
\addplot [name path = B, color=black, dotted, thick] coordinates {(0, 0.55) (0.45, 0.55)};
\addplot[name path=C, domain=0:1, samples=100, color=black, thick] coordinates {(0.5,0) (0.5,1)};
\addplot[domain=0:1, samples=100, color=black, thick, dotted] coordinates {(0.45,0) (0.45,1)};
\node at (axis cs:0.9,0.35) {$\frac{\l}{\m}[1-F(v)]$};
\pgfplotsset{ytick={0,1,0.5, 0.56, 0.3}, yticklabels={0,1,$\tilde v$, $\tilde v_0 +\d$, $\hat v_0$}} ;
\pgfplotsset{xtick={0,0.5}, xticklabels={0,1}};

\end{groupplot}
\end{tikzpicture}
\caption{Proof idea for \cref{thm:large-market}}
\label{fig:large-market}
\end{figure}

Analogously, to prove part (c), we consider the limit where the holding cost $d \to 0$. We construct a simple policy that caps inventory at $\bar{L} := \lfloor \sqrt{\mu/d} \rfloor$. When a buyer arrives, he is served at a price $\tilde{v}_0 + \delta$ if an item is available; otherwise, he is turned away. As $d \to 0$, the inventory cap $\bar{L}$ grows sufficiently large that the probability of a stockout vanishes. Meanwhile, the total expected holding cost scales with $d \bar{L} \approx \sqrt{d}$, which tends to zero. Consequently, the per-unit revenue of this simple policy converges to the oracle benchmark $R^*$, implying the same for the optimal policy.

Finally, these large-market results imply that No Allocation Delay (NAD) entails no loss of optimality in the limit. As the market grows large or frictions vanish, the revenue of any policy---even one that theoretically violates NAD by hoarding goods and buyers simultaneously---is bounded above by the oracle benchmark $R^*$. Since the optimal policy under NAD achieves this bound, the restriction becomes non-binding in the limit.

\section{Welfare maximization}\label{sec:welfare}

Thus far, our analysis has focused on revenue maximization. However, the framework readily extends to a social planner's problem of maximizing a weighted sum of revenue and buyer welfare. Suppose the planner assigns a weight of 1 to the seller's revenue and a weight $w \in [0,1]$ to the buyers' net surplus. The optimal policy for this objective is obtained simply by replacing the standard virtual value function $J(v)$ in the relaxed program $[\widetilde{\mathcal P}]$  with the weighted virtual value:\footnote{To derive this, we apply the standard envelope theorem to substitute payments in the weighted objective function. This yields an objective equivalent to maximizing the virtual surplus generated by $J_w(v)$, minus a term proportional to the lowest type's utility, $(1-w)U(0)$. For weights $w \le 1$, it is optimal to set the lowest type's utility to zero. For weights $w > 1$, the problem requires a seller participation constraint (e.g., non-negative revenue) to remain bounded; the solution then maximizes total welfare ($w=1$) and redistributes revenue to buyers.}
$$J_w(v) := v - (1-w)\frac{1-F(v)}{f(v)}.$$

Consequently, the welfare-maximizing policy retains the same threshold structure as the revenue-maximizing policy. In the case of pure welfare maximization ($w=1$), the virtual value reduces to the true valuation ($J_1(v)=v$). This eliminates the distortion arising from incentive constraints, fully aligning the admission policy with allocative efficiency.

\begin{prop}
    Given a Pareto weight of $w \in[0,1]$ on the buyers, welfare is maximized by replacing the thresholds of the optimal policy with $(\hat v_{-L^w},\ldots,\hat v_{-1}^w,\hat v_1^w, \ldots, \hat v_{K_w}^w)$, which are obtained by replacing $J(v)$ with $J_w(v)$.
\end{prop}

Given that the admission thresholds in our optimal policy play a role similar to that of a reserve price in a static auction, one might expect that placing greater weight on buyer welfare would lead to a lower threshold. Focusing on the case without storage ($d = \infty$), the following proposition states that this intuition holds only when the queue is relatively short. When the queue is already long, a social planner who places a larger weight on buyer utility is, in fact, \textit{more} selective in admitting a new buyer.

\begin{prop}[Comparative Statics in $w$] \label{prop:comp-stat-in-pareto-weight} Suppose $d=\infty$, and suppose the inverse hazard rate $(1-F(\cdot))/f(\cdot)$ is non-increasing. For any Pareto weights $w,w' \in [0,1]$ with $w > w'$, there exists $\bar k \in \{1,2,\ldots,\min\{K_{w'}, K_w\}\}$ such that $\hat v^w_k \leq \hat v^{w'}_k$ for $k \leq \bar k$, and $\hat v^w_k > \hat v^{w'}_k$ for $k \geq \bar k+1$.
\end{prop}
\cref{prop:comp-stat-in-pareto-weight} establishes a single-crossing property: for two weights $w>w'$, the thresholds are initially lower under $w$ when the queue is short, but they increase more quickly in queue length and may eventually overtake the thresholds under $w'$ (see \cref{fig:comp-stat-w}). 

To understand this result, we can compare the pure welfare maximizer ($w=1$, the Planner) and the revenue maximizer ($w=0$, the Seller). Two countervailing effects drive the single-crossing property.

\begin{figure}[htb]
	\begin{center}
		\includegraphics[width=10cm]{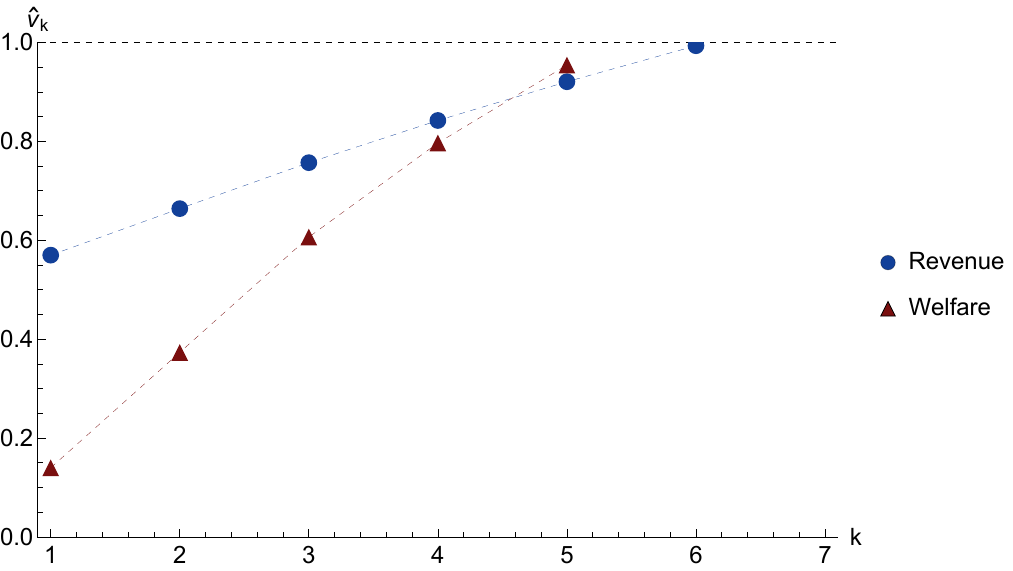}
	\end{center}
	{\footnotesize For $k=1,2,3,4$, the revenue-maximizing thresholds ($w=0$) are higher than the welfare-maximizing thresholds ($w=1$). For $k=5$, this is reversed.  \par}
		\caption{Comparative statics of thresholds in Pareto weight $w$ on buyers' utilities}\label{fig:comp-stat-w}
\end{figure}

The first is a {\it level} effect. The Planner values a match at the buyer's true valuation $v$, whereas the Seller values it at the virtual valuation $J(v)$. Since $J_1(v) = v > J_0(v)$ for all $v<1$, the Planner places a greater intrinsic value on the allocation than the Seller, who must deduct information rents from the social surplus. Consequently, the Planner is more concerned with a good being discarded than the Seller is. When the queue is empty or short, this effect dominates: the Planner lowers admission standards (e.g., $\hat{v}_1^{w=1} < \hat{v}_1^{w=0}$) to ensure the queue is nonempty, so that an arriving good is not wasted.

The second is a {\it slope} effect. While the Planner's valuation is higher in absolute terms, the Seller's  virtual value increases more rapidly in the buyer's type. Under the monotone hazard rate assumption, the marginal virtual value for the Seller ($J'_0(v)$) is strictly higher than that for the Planner ($J'_1(v)=1$). This is because admitting a higher-value buyer serves a dual purpose for the Seller: it not only secures a higher value for the potential transaction but also exerts greater competitive pressure, reducing the information rents accruing to other high-value buyers. This rent extraction motive makes the Seller marginally more eager to admit high-value types than the Planner, who views rents as neutral transfers. This effect encourages the Seller to be more permissive (i.e., to keep thresholds lower) as the queue grows.

Finally, a {\it dynamic feedback} effect reinforces the crossing. Recall that the admission of an additional buyer serves as ``insurance" for the future stream of surplus against the random arrival of goods. Because the Planner is more permissive when the queue is short (due to the level effect), the Planner's queue already provides substantial insurance against buyer stockouts. Having secured this safety buffer, the Planner has less need to admit marginal buyers as the queue lengthens. In contrast, the Seller, initially more restrictive, is not as insured as the Planner and must continue to add buyers to maintain competitive pressure. Together, these forces cause the Planner's thresholds to start lower but rise more rapidly, eventually overtaking the Seller's.

\section{Related Literature} \label{sec:lit-review} 
This paper integrates the classic mechanism design literature, exemplified by \cite{myerson1981}, with the rational queueing theory initiated by \cite{naor1969regulation}. While classic auction design typically assumes a static setting, our model captures the dynamic, stochastic nature of modern digital marketplaces. Essentially, we embed a Myersonian auction design problem within the dynamic framework of queueing systems.
Adapting standard mechanism design to this environment is non-trivial: the chosen mechanism is fully dynamic, compatible with the stochastic arrival of supply, and with allocations promised to both past and future agents. Methodologically, we provide a tractable linear programming (LP) formulation where we capture dynamic feasibility and stationarity with two new constraints, in addition to the standard IC and IR constraints. This approach provides a recipe that can be applied to steady-state mechanism design problems beyond the current setting.

 From the perspective of rational queueing theory, this paper applies a full-fledged mechanism design framework, optimizing all dimensions of system design---entry, exit, allocation, and information provision.\footnote{Several papers explore mechanism design with constraints on the designer's policy space.  Examples include \cite{naor1969regulation} and \cite{hassin1985}, and \cite{leshnoAER}, which study alternative queueing rules under complete information.  \cite{lingenbrink2019optimal} and \cite{Anunrojwong2020} study optimal information policy under first-come-first-served (FCFS) as a queueing rule. }  The closest antecedent is  \cite{che2021optimal}, who similarly adopt a design perspective on queue management. Our analysis departs from theirs in two key respects. First, we introduce heterogeneous, privately known preferences, whereas they assume homogeneous agents. This adds a screening dimension: the optimal policy must determine which agents should enter or remain in the queue based on their valuations. Second, we assume transferable utility, whereas they focus on a non-transferable utility setting. While transferable utilities allow the designer to internalize waiting costs through reimbursements, the waiting costs play a key role, driving the state-dependence of optimal reserve prices. Methodologically, we extend the steady-state approach of \cite{che2021optimal} to accommodate the infinite-dimensional state space inherent to environments with heterogeneous types.

Our work also relates to the literature on optimal screening in queues with heterogeneous preferences \citep{mendelson1990optimal, afeche2013incentive, afeche2016optimal, kittsteiner2005priority}. These papers also study mechanism design with transferable utilities, in which agents differ in their values for service, delay sensitivity, or processing time. The crucial distinction is that these mechanisms are typically ``static'': the screening rules (e.g., priority prices) do not condition on the current state of the queue. By contrast, we analyze a fully dynamic mechanism where the admission and allocation policies depend explicitly on the evolving profile of buyer types currently in the system.

This paper also relates to the literature on dynamic revenue management \citep{gallego1994optimal,gershkov2009dynamic,gershkov2009learning,pai2013optimal,board2016revenue,gershkov2018revenue}. These works typically study the allocation of a fixed inventory to buyers who arrive stochastically over a finite horizon (as with airline tickets sales). Consequently, the analysis is fundamentally non-stationary because the opportunity cost of allocation varies with remaining time and inventory. In contrast, our framework models a two-sided stochastic process where both buyers and goods arrive continuously. This necessitates a steady-state analysis, in which the opportunity cost is determined by the stationary distribution of the queue and inventory rather than by a terminal deadline.\footnote{More broadly, the literature has studied optimal dynamic mechanisms when agents receive private information across multiple periods \citep{bergemann2010dynamic, athey_efficient_2013,pavan_dynamic_2014}.}

Finally, we contribute to the growing literature on dynamic matching \citep{akbarpour2017thickness, 
  baccara2020optimal, leshno2019dynamic, doval2018efficiency, ashlagi2023matching}.  While these papers typically study non-transferable utility (NTU) environments, we consider fully transferable utilities (TU) and accommodate a broad spectrum of preference heterogeneity, akin to Myerson's classic work.\footnote{Within the TU domain, \cite{fershtman2022matching} also explore dynamic two-sided matching; however, their model focuses on agents who rematch every period, whereas we analyze the permanent allocation of goods to departing buyers.}

 \section{Conclusion} \label{sec:conclude}

 Motivated by applications such as cloud computing, gig platforms, and blockchain auctions, we have studied the design of optimal selling mechanisms in a dynamic environment where buyers and goods arrive asynchronously and stochastically. We derive the revenue-optimal dynamic mechanism at steady state by formulating and solving an infinite-dimensional linear programming problem. The resulting optimal policy is intuitive and implementable: it employs dynamic reserve prices that strictly increase with the queue length to manage competitive pressure, paired with efficient auctions that allocate goods to the highest-value buyers. This policy effectively smooths demand and supply over time, balancing the immediate gains from trade against the option value of future market thickness.

Our analysis paves the way for several extensions. First, one might generalize the service process beyond the M/M/1 structure, as considered by \cite{che2021optimal}.  Second, the framework could be extended to a {\it double auction setting} where goods are not merely ``arriving'' but are supplied by strategic agents with heterogeneous, privately known costs. Third, relaxing the assumption of homogeneous goods to allow for vertical or horizontal differentiation would be a significant step toward capturing the full complexity of platform exchanges. These extensions, among others, remain fruitful avenues for future research.

\bibliographystyle{economet}
\bibliography{reference}

\appendix

\section{ Proof of \cref{thm:relaxed-program,thm:relaxed-program-storable}} \label{sec:appendix-relaxed-program}

In this section, we first prove \cref{thm:relaxed-program-storable}, where we allow for $d<\infty$. \ref{subsec:service-scenario} then explains how the proof of \cref{thm:relaxed-program} follows if one sets $d=\infty$.

To prove \cref{thm:relaxed-program-storable}, we take a guess-and-verify approach. We rewrite the relaxed program as a linear program and set up the Lagrangian. We guess the binding set of constraints, and choose the multiplier for each binding constraint so that the coefficient of a relevant primal variable in the Lagrangian equals 0 (\ref{subsec:multipliers}). This gives us freedom in choosing the primal variables, and we choose them so that complementary slackness holds (\ref{subsec:choie-of-primal-variables}). Then, weak duality arguments will imply that our choice of primal variables is the optimal solution to the relaxed program (\ref{subsec:weak-duality}).

Here, we assume that there exist exogenous finite upper bounds $K,L\in \N$ such that the number of buyers in the queue at any given point in time cannot exceed $K$, and the number of goods in storage cannot exceed $L$. We choose $K$ and $L$ that are sufficiently large, in a sense we will make precise below.
In \cref{subsec:K=L=infty}, we show that imposing these upper bounds is without loss of optimality.

It will be convenient to work with $P_k(v)$ and $Q_\ell \:= \sum_{i=\ell}^{L} q_i$, rather than $p_k(v)$ and $q_\ell$.
The choice variables for the relaxed program can be represented by the tuple 
$$(\{P_k\}_{k\in\{1,\ldots,K\}},\{Q_\ell\}_{\ell\in\{1,\ldots,L\}}, \{y_k\}_{k\in\{1,\ldots,K\}},\{z_\ell\}_{\ell\in\{1,\ldots,L\}},X(v) ),$$
where each $P_k$ is a CDF, each $Q_\ell$ is a real number in $[0,1]$, and each $y_k$, each $z_\ell$, and $X$ are measurable functions from $[0,1]$ to itself.
The relaxed program is
\begin{align*}
    [\widetilde{\mathcal{P}}] \quad \max_{P,Q,y,z,X} &\, \l \int_0^1 J(v)X(v)f(v)dv - c \sum_{k=1}^K k[P_{k+1}(0)-P_k(0)]  - d \sum_{\ell=1}^L \ell (Q_\ell - Q_{\ell+1}) \\
    \text{subject to } &\,(\widetilde{RF}),\, (B_k)\,\,\forall k, \mbox{ and } (B_{-\ell})\,\,\forall \ell.
\end{align*}
We linearize the program via the change of variables $Y_k(v) := y_k(v)(P_{k+1}(v) - P_k(v))$ and $Z_\ell(v) := (Q_\ell - Q_{\ell+1})f(v)z_\ell(v)$. Assigning multipliers $\alpha, \beta, \gamma, \eta,$ and $\kappa$ to the constraints $(\widetilde{RF})$, $(B_k)$, $(B_{-\ell})$ and the bounds on $Y_k$ and $Z_\ell$, and applying integration by parts, the Lagrangian $\ms L$ becomes
\begin{align}
    & \int_0^1 \l \left(J(v) - \int_0^v \alpha(s)ds\right) X(v)f(v)dv \label{eq44} \\
    & + \sum_{k=1}^K \int_0^1 \bigg( \mu\alpha(v) - \mu\beta_k(v) - \eta_k(v)  \bigg) Y_k(v)dv \label{eq45} \\
    & + \sum_{k=1}^K \int_0^1 \bigg( \l[1-F(v)](\beta_k(v) - \beta_{k+1}(v)) - (\eta_k(v)-\eta_{k-1}(v)) \bigg) (P_k(v) - P_k(0)) dv \label{eq20} \\
    & + \sum_{k=2}^K \left(\int_0^1 \Big(\l[1-F(v)](\beta_k(v) - \beta_{k+1}(v)) - (\eta_k(v)-\eta_{k-1}(v)) \Big) dv + c \right)  P_k(0) \label{eq22}\\
    & + \left(\int_0^1 \Big( \l[1-F(v)](\beta_1(v) - \beta_2(v)) - \eta_1(v)   \Big) dv + \mu \gamma_1+ c \right)  P_1(0) \label{eq24}\\
    & + \left( \int_0^1 \Big( - \l[1-F(v)]\beta_1(v)  + \kappa_1(v) f(v) \Big) dv - \mu(\gamma_1 - \gamma_2) - d\right) Q_1 \cr
    & + \sum_{\ell=2}^L \left( \int_0^1 \Big( \kappa_\ell(v) - \kappa_{\ell-1}(v)\Big) f(v)   dv - \mu(\gamma_\ell - \gamma_{\ell+1}) - d\right) Q_\ell \cr
    & + \sum_{\ell=1}^L \int_0^1 \left( \l\int_0^v\alpha(s)ds - \l\gamma_\ell - \kappa_\ell(v) \right) Z_{\ell} (v) dv \label{eq28}\\
    & + \int_0^1 \eta_K(v) dv - cK, \notag
\end{align}
where we have defined $\beta_0(v) = \beta_{K+1}(v) = \eta_0(v) = \eta_{-1}(v) = \gamma_0 = \gamma_{L+1} = \kappa_0(v) = \kappa_{-1}(v) = 0$, $P_0(0):=Q_1$, and $Q_0:=P_1(0)$. 

\subsection{Choice of Multipliers}\label{subsec:multipliers}

Take as given for now threshold values $\{\hat v_k\}_{k=1,\ldots,K}$ and $\{\hat v_{-\ell}\}_{\ell=1,\ldots,L}$ satisfying $\hat v_K \geq \ldots \geq \hat v_2 \geq \hat v_1 \geq \hat v_{-1} \geq \hat v_{-2} \geq \ldots \geq \hat v_{-L} \geq J^{-1}(0).$
We choose the Lagrangian multipliers as follows:
\begin{align}
    &\,\alpha(v) = \begin{cases}
        J'(v) & \text{if } v \geq J^{-1}(0); \\
        0 & \text{if } v < J^{-1}(0),
    \end{cases} \notag \\
    &\, \beta_k(v) = \begin{cases}
        \dfrac{\sigma(v)^k + \sigma(v)^{k+1} + \cdots + \sigma(v)^i}{1 + \sigma(v) + \cdots + \sigma(v)^{i}}J'(v) &\text{if } v\in [\hat v_{i},\hat v_{i+1}),\, k\in\{1,\ldots,i\}; \\
        0 & \text{if } v\in [\hat v_{i},\hat v_{i+1}),\, k=\{i+1, \ldots, K\},
                \end{cases} \label{beta} \\
    &\, \gamma_\ell = J(\hat v_{-\ell}), \notag\\
     &\,\eta_k(v) = \mu\alpha(v) - \mu\beta_k(v) \qquad \forall k\geq 1, \label{eq23}\\
    &\,\kappa_\ell(v) = \begin{cases}
        \l (J(v) - \gamma_\ell) & \text{if } v \geq \hat v_{-\ell};\\
        0 & \text{if } v<\hat v_{-\ell},
    \end{cases} \notag
\end{align}
where  $\sigma(v)\:= \mu/\l[1-F(v)]$.

\paragraph*{Choice of Thresholds.}

First, we take the goods thresholds $\{\hat v_{-\ell}\}$ as given and determine the buyer thresholds $\{\hat v_k\}$. We will later jointly determine both types of thresholds.

We choose the thresholds $\hat v_k$ iteratively, using the following algorithm.
\begin{enumerate}%
    \item Define
        \begin{align}
        H_k(\hat v_{k-1},\hat v_k)\:= - \m \int_{\hat v_{k-1}}^{\hat v_k} \frac{\sigma(v)^{k-1}}{1 + \sigma(v) + \cdots + \sigma(v)^{k-1}} J'(v) d v + c. \label{eq17}
     \end{align}
     Let
     \begin{align}
    \hat v_1 = \begin{cases}
        J^{-1}\left(\gamma_1 + c/\mu\right) \caseif \gamma_1 + c/\mu \leq 1 \\
        1 \caseif \text{otherwise.}
    \end{cases}\label{eq33}
    \end{align}
     Set $k=2$.
    \item If $H_k(\hat v_{k-1},1)\leq 0$, go to step (c); otherwise, go to step (d).
    \item Set $\hat v_k \in (\hat v_{k-1},1]$ so that $H_k(\hat v_{k-1}, \hat v_k) = 0$.\footnote{Such a $\hat v_k$ is unique because the integrand is strictly positive.} Increase $k$ by 1 and go to step (b). 
    \item Set $\hat v_k = \ldots = \hat v_K = 1$, and terminate.
\end{enumerate}
The algorithm terminates after a finite number of steps, since $\sigma(v)^{k-1} / (1 + \sigma(v) + \cdots + \sigma(v)^{k-1}) \in [0,1]$. The initial threshold $\hat v_1$ depends on $\g_1 = J(\hat v_{-\ell})$, which for now is taken as given. The expression $H_k$ represents the coefficient of $P_k(0)$ in the Lagrangian, so setting $H_k=0$ will allow us (in \ref{subsec:choie-of-primal-variables}) to choose the primal variables to maximize the Lagrangian and satisfy complementary slackness.

If $\hat v_1 = 1$, define $K^* = 0$, and otherwise, define $K^*$ to be the largest positive integer $k$ such that $\hat v_k<1$. $K^*$ will represent the maximum number of buyers that may be in the queue under the optimal policy. Note that the value of $K^*$ does not depend on the exogenous upper bound on queue length, $K$, as long as $K \geq K^*$. Hence, we could have chosen $K$ to be any integer greater than $K^*$.

It remains to determine the goods thresholds $\{\hat v_{-\ell}\}$. To encode the dependency of $\b_1$ on $\g_1 = J(\hat v_{-1})$, let us write $\b_1(v;\hat v_{-1})$. Define
\begin{align*}
A(\hat v_{-1}, \hat v_{-2})    := &\, -\int_0^1 \l[1-F(v)]\beta_1(v;\hat v_{-1})dv \\
    &\, + \int_{\hat v_{-1}}^1 \l(J(v) - J(\hat v_{-1})) \ddd F(v) - \mu(J(\hat v_{-1}) - J(\hat v_{-2})) -d\\
    B(\hat v_{-\ell+1}, \hat v_{-\ell}, \hat v_{-\ell-1}) := &\,\l[J(\hat v_{-\ell+1}) - J(\hat v_{-\ell})](1-F(\hat v_{-\ell+1})) \\
    &\, + \int_{\hat v_{-\ell}}^{\hat v_{-\ell+1}} \l(J(v)-J(\hat v_{-\ell})) \ddd F(v) - \mu(J(\hat v_{-\ell}) - J(\hat v_{-\ell-1})) - d.
\end{align*}
The expression $A$ represents the coefficient of $Q_1$ in the Lagrangian, whereas $B$ is the coefficient of $Q_\ell$ for $\ell \geq 2$.  
We choose a nonnegative integer $L^* \leq L$ and thresholds $\{\hat v_{-\ell}\}$ such that $v_\ell = J^{-1}(0)$ for all $\ell > L^*$, and the coefficient of $Q_\ell$ is 0 for $\ell\leq L^*$ and nonpositive for $\ell > L^*$. In \cref{subsec:IVT}, by iteratively appealing to the intermediate value theorem, we show that there exist such an $L^*$ and $\{\hat v_{-\ell}\}$.
By choosing the goods thresholds in this way, we will (in \ref{subsec:choie-of-primal-variables}) be able to choose the primal variables to maximize the Lagrangian and satisfy complementary slackness. 

The value $L^*$ represents the maximum number of goods that may be stored under the optimal policy. Analogously to the buyers' side, it would have been enough to set the exogenous upper bound $L$ on inventory size to be any integer greater than $L^*$.

Note that the buyer thresholds $\{\hat v_k\}$ and the goods thresholds $\{\hat v_{-\ell}\}$ are determined jointly. The choice of $\hat v_{-1}$ affects the values of the positive thresholds $\hat v_k$'s. This affects $\beta_1(v;\hat v_{-1})$, which in turn enters the expression $A(\hat v_{-1},\hat v_{-2})$ and thus affects the choice of negative thresholds.

\subsection{Choice of Primal Variables}\label{subsec:choie-of-primal-variables}

We now choose the primal variables $(P, Q, y, z, X)$. 
First, we choose
\begin{align*}
    y_k(v) = 1 \qquad \forall v\in[0,1], \,\, \forall k=1,\ldots,K.
\end{align*}
Given our choice of $\eta$, the coefficient of $Y_k(v)$ in \eqref{eq45} is 0, so the Lagrangian is trivially maximized with respect to $y$. Moreover, the constraint $Y_k(v) \leq P_{k+1}(v) - P_k(v)$ holds with equality, so complementary slackness\footnote{In what follows, we use the term \textit{complementary slackness} to mean that (i) the constraint is satisfied, and (ii) the constraint is satisfied with equality if the corresponding multiplier is strictly positive.} with respect to the multiplier $\eta$ is satisfied. Next, we choose
\begin{align*}
    z_\ell(v) = \mathbb{I}\{ v \geq \hat v_{-\ell} \} \qquad \forall \ell = 1, \ldots, L.
\end{align*}
Given our choice of $\kappa$, the coefficient of $Z_\ell(v)$ in \eqref{eq28} is equal to 0 when $v \geq \hat v_{-\ell}$ and nonpositive otherwise, so the Lagrangian is maximized with respect to $z$. Also, complementary slackness with respect to $\kappa$ is satisfied.

\subsubsection{Choice of $P$}\label{subsubsec:choice-of-P}

Here, we take $Q_1$ as given and pin down $\{P_k\}$ up to $Q_1$. Next, in \ref{subsubsec:choice-of-Q}, we will take $P_1(0)$ as given and pin down $\{Q_\ell\}$ up to $P_1(0)$. Then, we will argue that there exist $P_1(0)$ and $Q_1$ that are consistent with each other.

We wish to choose $P_k(v)$ so that the Lagrangian is maximized, and complementary slackness holds. 
Our choice of multipliers, along with the algorithm described in \ref{subsec:multipliers} for choosing the buyer thresholds, implies that the coefficient of $P_k(0)$ in the Lagrangian is positive for $k>K^*$ and equals 0 for $k \leq K^*$.\footnote{To see this, note that the coefficient of $P_1(0)$ in \eqref{eq24} reduces to $-\mu J(\hat v_1) + \mu \gamma_1+ c$, while the coefficient of $P_k(0)$ for $k \geq 2$ in \eqref{eq22} simplifies to
\begin{align}
    \int_{\hat v_{k-1}}^{\hat v_k} -\mu\beta_{k-1}(v) dv + c = H_k(\hat v_{k-1},\hat v_k).\notag
\end{align}}
Likewise, the multipliers $\b$ and $\eta$ defined in \eqref{beta} and \eqref{eq23} ensure that the coefficient of $P_k(v) - P_k(0)$ is weakly negative when $v<\hat v_k$ and equal to 0 when $v\geq \hat v_k$; \cref{sec:appendix-multipliers} explains how the multipliers $\b$ and $\eta$  satisfy these properties. 
Hence, we set $P_k(v)=1$ for all $k>K^*$ and all $v\in [0,1]$, and we set $P_k(v) = P_k(0)$ for all $k\leq K^*$ and $v<\hat v_k$; this ensures that the Lagrangian is maximized with respect to $P_k(v)$. For the remaining case of $k\leq K^*$ and $v \geq \hat v_k$, we will choose $P_k(v)$ to satisfy complementary slackness.

Since the multiplier $\b_k(v)$ is positive only for $v \geq \hat v_k$, to satisfy complementary slackness with respect to $\b$, it is enough to ensure that $(B_k)$ holds, and that it holds with equality for each $v \geq \hat v_k$ and for each $k \leq K^*$.
Consider some $i\in\{1,\ldots,K^*\}$ and $v\in [\hat v_i, \hat v_{i+1})$. 
We want $(B_k)$ to bind for $k=1,\ldots,i$. Since we chose $y_k(v) = 1$, this means that
\begin{align}
    \l[1-F(v)] (P_k(v) - P_{k-1}(v)) = \m (P_{k+1}(v) - P_k(v)) \qquad \forall k\in\{1,\ldots,i\}, \,  \forall v \in [\hat v_i,\hat v_{i+1}). \notag
\end{align}
Defining $\rho(v)\:=\l[1-F(v)]/\m$ and $P_0(v)\:=Q_1$, we have
\begin{align}
    P_k(v) - P_{k-1}(v) = \rho(v)^{k-1}(P_1(v) - Q_1) \qquad \forall k\in\{1,\ldots,i+1\}, \,  \forall v \in [\hat v_i,\hat v_{i+1}). \label{eq50}
\end{align}
Summing the above equation across $k$, the above equation gives
\begin{align}
    P_k(v) = Q_1 + (1 + \rho(v) + \cdots + \rho(v)^{k-1})(P_1(v) - Q_1) \qquad \forall k\in\{1,\ldots,i+1\}, \,  \forall v \in [\hat v_i,\hat v_{i+1}). \label{eq47}
\end{align}
In particular, for $k=i+1$, we have
\begin{align}
    P_{i+1}(0) = P_{i+1}(v) = Q_1 + (1+ \rho(v) +\cdots + \rho(v)^i) (P_1(v) - Q_1) \qquad \forall v \in [\hat v_i,\hat v_{i+1}). \label{eq46}
\end{align}
Let us suppose for now that $Q_1$ is exogenously given. Then, for each $v \in [\hat v_i,\hat v_{i+1})$, choosing $P_{i+1}(0)$ pins down $P_1(v)$ via \eqref{eq46}, which in turn pins down $P_k(v)$ for each $k\in \{1,\ldots,i\}$, via \eqref{eq47}. Hence, the choice of $P_{i+1}(0)$ for each $i \in \{1,\ldots,K^*\}$ pins down $P_k(v)$ for all $k\in\{1,\ldots,K\}$ and all $v \geq \hat v_1$.

We wish to choose each $P_{i+1}(0)$ (and also $P_1(0)$) so that $P_1$, and therefore $P_k$ for each $k\in\{1,\ldots,K\}$, are continuous. In particular, we want $P_k$ not to have atoms at the thresholds $\{\hat v_k\}$.
This will imply that $P_k$ is constant on the closed interval $[0,\hat v_k]$. Since the balance inequality $(B_k)$ holds for $v=\hat v_k$, it will then also hold for $v = [0,\hat v_k)$,\footnote{For $v \in [0,\hat v_k)$, the LHS of the balance inequality, $\l[1-F(v)](P_k(v) - P_{k-1}(v))$, is decreasing in $v$, while the RHS, $\m(P_{k+1}(v) - P_k(v))$, is constant in $v$.} guaranteeing that our choice of $P$ will satisfy $(B_k)$.\footnote{The continuity of $P_k$ will also be used when we show that our choice of $P_k$ can be implemented by the optimal policy (\cref{thm:optimal-policy,thm:optimal-policy-storable}).}

For $P_1$ given by \eqref{eq46} to be continuous at $\hat v_i$, it must be that $P_1(\hat v_i) =\lim_{v \uparrow \hat v_i} P_1(v)$, i.e.
\begin{align}
    \frac{P_{i+1}(0) - Q_1}{1+ \rho(\hat v_i) + \cdots + \rho(\hat v_i)^i} = \frac{P_i(0) - Q_1}{1+\rho(\hat v_i) + \cdots + \rho(\hat v_i)^{i-1}}.\notag
\end{align}
That is, for each $i\in\{1,\ldots,K^*\}$, $P_i(0)$ is determined iteratively by $P_{i+1}(0)$, and the initial condition is $P_{K^*+1}(0) = 1$. In particular, we have
\begin{align}
    P_1(0)-Q_1 &\, = (P_2(0)-Q_1)\frac{1}{1+\rho(\hat v_1)} \cr
    &\, = (P_3(0) - Q_1) \frac{1}{1+\rho(\hat v_1)}\frac{1+ \rho(\hat v_2)}{1+\rho(\hat v_2) + \rho(\hat v_2)^2} \cr
    &\, \cdots \cr
    &\, = (1-Q_1) \frac{1}{1+\rho(\hat v_1)} \cdots \frac{1+\rho(\hat v_{K^*}) + \cdots + \rho(\hat v_{K^*})^{K^*-1}}{1+\rho(\hat v_{K^*}) + \cdots + \rho(\hat v_{K^*})^{K^*}} \cr
    &\, =: (1-Q_1) \Delta_{K^*}. \label{eq48}
\end{align}

\subsubsection{Choice of $Q$}\label{subsubsec:choice-of-Q}

As we did for $P$, we wish to choose $Q_\ell$ so that the Lagrangian is maximized, and complementary slackness holds. Given our choice of multipliers and thresholds in \ref{subsec:multipliers}, the coefficient of $Q_\ell$ is 0 for $\ell \leq L^*$ and is nonpositive for $\ell > L^*$. Therefore, we choose $Q_\ell = 0$ for $\ell > L^*$; this guarantees that the Lagrangian is maximized with respect to $Q_\ell$.

Since the multiplier $\g_\ell = J(\hat v_{-\ell})$ is positive only for $\ell \leq L^*$, to satisfy complementary slackness with respect to $\g_\ell$, it is enough to ensure that $(B_{-\ell})$ holds, and that it holds with equality for $\ell \leq L^*$. Choosing $Q_\ell = 0$ for $\ell > L^*$ guarantees that $(B_{-\ell})$ holds for $\ell > L^*$. For $\ell \leq L^*$,  let us find $Q_\ell$ that makes $(B_{-\ell})$ bind, i.e.
\begin{align}
    \m (Q_{\ell-1} - Q_\ell) = \l[1-F(\hat v_{-\ell})](Q_\ell - Q_{\ell+1}) \qquad \forall \ell\leq L^*. \notag
\end{align}
Using the definition $\sigma(v) := \m / \l[1-F(v)]$ we have
\begin{align}
    Q_{\ell-1} - Q_\ell &\, = \sigma(\hat v_{-\ell+1}) \sigma(\hat v_{-\ell+2}) \cdots \sigma(\hat v_{-1})(P_1(0)-Q_1) \qquad \forall \ell \in \{2,\ldots, L^*+1\}, \notag \\
    Q_0 - Q_1 &\, = P_1(0) - Q_1,
\end{align}
where the second line follows by the definition $Q_0\:= P_1(0)$.
Taking a summation gives
\begin{align*}
    Q_\ell %
    &\, = P_1(0) - W_{\ell-1}(P_1(0) - Q_1),
\end{align*}
where $W_{\ell-1}\coloneqq 1 + \sigma(\hat v_{-1}) + \sigma(\hat v_{-1})\sigma(\hat v_{-2}) + \cdots + \sigma(\hat v_{-1})\sigma(\hat v_{-2}) \cdots \sigma(\hat v_{-\ell+1})$ for $\ell \geq 2$ and $W_0\:=1$.
In particular, we have $0 = Q_{L^*+1} = P_1(0) - W_{L^*} (P_1(0) - Q_1)$, i.e.
\begin{align}
    (W_{L^*}-1)P_1(0) =  W_{L^*}  Q_1. \label{eq49}
\end{align}
One can see that there exist $P_1(0),Q_1 \in [0,1]$ satisfying \eqref{eq48} and \eqref{eq49}. As we have shown above, choosing $P_1(0)$ and $Q_1$ then pins down all of $P_k(v)$ and $Q_\ell$.

\subsubsection{Choice of $X$}
\label{subsubsec:choice-of-X}

Because we assumed that the density $f$ is absolutely continuous, $J(v)$ is also absolutely continuous, so we may write $J(v) = \int_0^v \alpha(s)ds$ for $v \geq J^{-1}(0)$. Hence, the coefficient of $X(v)$ in \eqref{eq44} is nonpositive when $v< J^{-1}(0)$ and equal to 0 when $v \geq J^{-1}(0)$. Therefore, we choose $X(v) = 0$ for $v< J^{-1}(0)$; this ensures that the Lagrangian is maximized with respect to $X$.

Since the multiplier $\a(v)$ is positive only for $v \geq J^{-1}(0)$, to satisfy complementary slackness with respect to $\a$, it is enough to ensure that the reduced-form feasibility constraint $(\widetilde{RF})$ holds, and that it holds with equality for $v \geq J^{-1}(0)$. Both the LHS and the RHS of the constraint are equal to 0 when $v=1$, so to make $(\widetilde{RF})$ bind for $v \geq J^{-1}(0)$, it is enough to choose $X$ so that the derivatives with respect to $v$ of both sides are equal to each other for $v \geq J^{-1}(0)$. That is, we want
\begin{align}
    \l X(v) f(v) = \m \frac{\dd P_1(v)}{\dd v} + \l \sum_{\ell=1}^{L^*} (Q_\ell - Q_{\ell+1}) f(v) z_\ell(v). \notag
\end{align}
Therefore, if $v \geq \hat v_1$, then we choose
\begin{align}
    X(v) = \frac{\m}{\l f(v)} 
\frac{\dd P_1(v)}{\dd v} + Q_1, \notag
\end{align}
and if $v\in [\hat v_{-\ell}, \hat v_{-\ell+1})$ for $\ell \in \{2,\ldots,L^*\}$, then we choose $X(v) = Q_\ell$.
If $v \in [\hat v_{-1}, \hat v_1)$, then we choose $X(v) = Q_1$. If $v \in [J^{-1}(0), \hat v_{-L^*})$, then we choose $X(v) = 0$. 

Clearly, $(\widetilde{RF})$ holds as well for $v < J^{-1}(0).$

\subsection{Weak Duality}\label{subsec:weak-duality}

Let $(P^*, Q^*, y^*, z^*, X^*)$ denote the primal variables that we have constructed. Since $(P^*, Q^*, y^*, z^*, X^*)$ maximizes the Lagrangian given the constructed multipliers and satisfies complementary slackness, by standard weak duality arguments, $(P^*, Q^*, y^*, z^*, X^*)$ solves the relaxed program $[\tilde P]$.

\subsection{Service Scenario ($d=\infty$)} \label{subsec:service-scenario}

When $d=\infty$, the primal and dual variables for goods storage all collapse. That is, we have $Q_\ell = 0$, $\g_\ell = 0$ and $\hat v_{-\ell} = J^{-1}(0)$ for all $\ell \in \{1,\ldots,L\}$. In particular, because $\g_1=0$, we obtain an iterative characterization of the buyer thresholds $\{\hat v_k\}$ via the algorithm described in \cref{subsec:multipliers}, and there is no need to appeal to intermediate value theorems as in the storable goods scenario ($d<\infty)$.

\section{Proof of \cref{thm:optimal-policy} and \cref{thm:optimal-policy-storable}}
\label{sec:appendix-proof-of-thm:optimal-policy}

We now prove \cref{thm:optimal-policy-storable}. (\cref{thm:optimal-policy} will follow from this by setting  $d=\infty$.) 
To this end, we first show that the policy $\phi^*$ implements the stationary distributions $P^*$ and $Q^*$. We then show that $\phi^*$, together with $P^*$ and $Q^*$, implements $y^*$, $z^*$, and $X^*$. 
Finally, we argue that $\phi^*$ is incentive compatible.

\paragraph*{Implementing  $P^*$ and $Q^*$.}  
We will first show that $P^*$ and $Q^*$ are stationary under $\phi^*$. Using this, we will then argue that $\phi^*$ must implement a unique (joint) stationary distribution over $\{(\bv;\ell)_t\}_{t\geq 0}$, which has the marginals $P^*$ and $Q^*$.

Let us show that $P^*$ is stationary under $\phi^*$, i.e. that the (instantaneous) inflow and outflow from the set $S_k(v) = \{\mathbf{v} : v_k \le v\}$ are equal, for all $k$ and $v$.\footnote{Recall that $P^*_k(v)$ is the probability of $S_k(v)$ under the optimal stationary distribution.} 
For $k>K^*$, we have $P^*_k(v)=1$ for all $v$, and there is no outflow from $S_k(v)$, so stationarity holds trivially. Consider $k \leq K^*$.
For $v \geq \hat v_k$, the (actual) outflow from $S_k(v)$ is 
$$ \l[1-F(v)] (P^*_k(v) - P^*_{k-1}(v)),  $$
while the (actual) inflow is
$$ \m (P^*_{k+1}(v) - P^*_k(v)). $$
Since we chose $P^*$ so that the balance inequality $(B_k)$ binds for $v \geq \hat v_k$, the above two expressions are equal.  

Now, consider $v < \hat v_k$. Recall that, under the optimal policy $\phi^*$, there can be $k$ buyers in the queue only if all of them have values weakly above $\hat v_k$. 
Hence, under $\phi^*$, for $v < \hat v_k$, an outflow occurs from $S_k(v)$ if and only if there are exactly $k-1$ buyers whose values are \textit{weakly} above $\hat v_k$, and a new buyer arrives with value \textit{weakly} above $\hat v_k$. However, because we chose each $P^*_k(v)$ so that there is no atom at $\hat v_k$, and the distribution $F$ of an arriving buyer's type is continuous, the quantifiers ``weakly'' in the previous sentence may be replaced with ``strictly''; that is, the outflow can be written as\footnote{Recall that $\l[1-F(\hat v_k)]$ is the rate at which a buyer arrives with value \textit{strictly} above $\hat v_k$, whereas $P^*_k(\hat v_k) - P^*_{k-1}(\hat v_k)$ is the probability that there are exactly $k-1$ buyers with values \textit{strictly} above $\hat v_k$.} 
$$ \l[1-F(\hat v_k)] (P^*_k(\hat v_k) - P^*_{k-1}(\hat v_k)).  $$
This is equal to the outflow from $S_k(\hat v_k)$.  

Likewise, under the optimal policy, for $v < \hat v_k$, an inflow into $S_k(v)$ occurs if and only if there are exactly $k$ buyers whose values are \textit{weakly} above $\hat v_k$, and an item arrives. Again, by the continuity of $P_k^*$, this inflow is
$$ \m (P^*_{k+1}(\hat v_k) - P^*_k(\hat v_k)),$$
which is equal to the inflow to $S_k(\hat v_k)$. Since inflow and outflow are equal for $S_k(\hat v_k)$, they must be equal for $S_k(v)$ as well. We have thus proved that $P^*_k(v)$ is stationary under the policy $\phi^*$.

Next, we show that $Q^*$ is stationary under $\phi^*$. For $\ell > L^*$, we have $Q^*_\ell = 0$, and there is no inflow, so stationarity holds trivially. For $\ell \leq L^*$, the outflow from $S_{-\ell}$ is
$$\l[1-F(\hat v_{-\ell})](Q^*_\ell - Q^*_{\ell+1}),$$
while the inflow is
$$ \m (Q^*_{\ell-1} - Q^*_\ell). $$
Since we chose $Q^*$ so that the balance inequality $(B_{-\ell})$ binds for $\ell \leq L^*$, the above two expressions are equal. We have thus proved that $Q^*_\ell$ is stationary under $\phi^*$.

Finally, since the empty state $(\mathbf{0};0)$ has a finite mean return time under $\phi^*$, the policy $\phi^*$ induces a positive recurrent regenerative process $\{(\bv;\ell)_t\}_{t\geq 0}$, which has a unique stationary distribution \citep{asmussen2003applied,THORISSON1992237}. Moreover, any projection of $\{(\bv;\ell)_t\}_{t\geq 0}$ is also a regenerative process (\cite{asmussen2003applied} page 169, Proposition 1.1) and thus has a unique stationary distribution; this implies that, under $\phi^*$, the CDF $P^*_k(v)$ corresponds to the unique stationary marginal distribution of the $k$-th order statistic, and likewise for $Q^*_\ell$. Therefore, the unique stationary distribution of the process $\{(\bv;\ell)_t\}_{t\geq 0}$ induced by $\phi^*$ must have the marginals $P^*$ and $Q^*$.

\paragraph*{Implementing $y^*$ and $z^*$.}

When buyers are waiting, the policy $\phi^*$ always assigns an incoming good to the highest-value buyer, so it clearly implements $y^*_k(v)=1$ for all $v$ and $k$. Also, when $\ell$ goods are stored, the policy assigns a good to an arriving buyer if and only if his value is above $\hat v_{-\ell}$, so $z_\ell^*(v)$ is implemented.

\paragraph*{Implementing $X^*$.}
Let $\tilde X(v)$ denote the interim allocation probability for a type $v$-buyer that is induced by the policy $\phi^*$, under the stationary distribution $(P^*,Q^*)$. We wish to show that $\tilde X(v) = X^*(v)$ for all $v\in[0,1]$, where $X^*(v)$ is the solution to the relaxed program derived in \cref{subsubsec:choice-of-X}.

For $\ell \in \{1, \ldots, L\}$, let $X_{-\ell}(v)$ denote the interim allocation probability under the optimal policy of a type-$v$ buyer who just arrived, conditional on there being exactly $\ell$ goods stored. Then, we have
\begin{align*}
    X_{-\ell}(v) = \begin{cases}
        1 \caseif v \geq \hat v_{-\ell} ;\\
        0 \caseif v < \hat v_{-\ell}.
    \end{cases}
\end{align*}
Hence, for $v \leq \hat v_1$, we obtain $\tilde X(v)$ as follows:
\begin{align*}
    \tilde X(v) = \begin{cases}
        0 & \text{if } v< \hat v_{-L^*}; \\
        Q_\ell & \text{if } v \in [\hat v_{-\ell}, \hat v_{-\ell+1}) \text{ for } \ell \in \{2,\ldots,L^*\};\\
        Q_1 \caseif v \in [\hat v_{-1}, \hat v_{1}),
    \end{cases}
\end{align*}
which coincides with $X^*(v)$.

For $k\in \{1,\ldots,K^*\}$, let $X_k(v)$ denote the interim allocation probability, under the optimal policy $\phi^*$, of a type-$v$ arriving buyer when his value is $k$-th highest in the queue.If $v < \hat v_1$, then $X_k(v)=0$ for any $k$.
Suppose the arriving buyer has value $v \in [\hat v_i, \hat v_{i+1})$, for some $i\in \{1,\ldots,K^*\}$.\footnote{We define $\hat v_{k}\:=1$ for all $k \geq K^*+1$.} Clearly, $X_k(v) = 0$ for $k \geq i+1$, since the buyer will not be allowed to join the queue. When $k \leq i$, the buyer joins the queue. Once the buyer joins the queue, every time a good arrives, the buyer's ranking improves by 1, and every time a new buyer arrives with value greater than $v$, the buyer's ranking drops by 1. The buyer obtains the good if his ranking reaches 0, and the buyer leaves the queue if his ranking reaches $i+1$. Therefore, for each $v \in [\hat v_i, \hat v_{i+1})$, the conditional interim allocation $X_k(v)$ is a solution to the gambler's ruin problem and is given by\footnote{The interim allocation probabilities are a solution to the system of difference equations $X_k(v) = \frac{1}{1+\rho(v)}X_{k-1}(v) + \frac{\rho(v)}{1+\rho(v)}X_{k+1}(v)$.}
\begin{align}
    X_k(v) = \begin{cases}
    \dfrac{k}{i+1} & \text{if } k \leq i \text{ and } \mu = \l[1-F(v)]\\
    \dfrac{\rho(v)^{i-k+1} - 1}{\rho(v)^{i+1}-1} & \text{if } k \leq i \text{ and } \mu\neq \l[1-F(v)]  \\
    0    & \text{if } k \geq i+1.
    \end{cases}\label{eq:X_k}
\end{align}
Suppose $\mu\neq \l[1-F(v)]$. Then, the (unconditional) interim allocation probability $\tilde X(v)$ for a buyer with value $v \in [\hat v_i, \hat v_{i+1})$, for $i \in \{1,\ldots,K^*\}$, is given by
\begin{align*}
    \tilde X(v) = &\, Q_1 + \sum_{k=1}^{K^*} [P_k(v) - P_{k-1}(v)] X_k(v) \\
    = &\, Q_1 + \sum_{k=1}^i \rho(v)^{k-1} \frac{P_{i+1}(0) - Q_1}{1+\rho(v) + \cdots + \rho(v)^i} \dfrac{\rho(v)^{i-k+1} - 1}{\rho(v)^{i+1}-1} \\
    = &\, Q_1 + (P_{i+1}(0) - Q_1) \frac{i(\rho(v)^{i+1}-\rho(v)^i) +1 - \rho(v)^i}{(1-\rho(v)^{i+1})^2}.
\end{align*} 
where the second equality uses \eqref{eq50}, \eqref{eq46}, and the definition of $X_k(v)$. On the other hand,
For $i\in \{1,\ldots,K\}$ and $v \in [\hat v_i,\hat v_{i+1})$, the solution to the relaxed program is
\begin{align*}
    X^*(v) &\, = \frac{\m}{\l f(v)} 
\frac{\dd P_1(v)}{\dd v} + Q_1 \\
&\, = \frac{1+2\rho(v) + 3\rho(v)^2 +\cdots + i\rho(v)^{i-1}}{(1+\rho(v) +\cdots+\rho(v)^i)^2} (P_{i+1}(0) - Q_1) + Q_1\\
    &\, = \tilde X(v).
\end{align*}
Likewise, when $\mu = \l[1-F(v)]$, we have
\begin{align*}
    \tilde X(v) = X^*(v) = \frac{i}{2(i+1)} (P_{i+1}(0) - Q_1) + Q_1.
\end{align*}

\paragraph*{Incentives.}
Finally, we check that the policy $\phi^{*}$ satisfies obedience and incentive compatibility of buyers. Obedience is clearly satisfied, since buyers are reimbursed for their ex post waiting cost, so they never benefit by leaving the queue against the policy's recommendation. To check incentive compatibility, it is enough to verify that $X^*$ is non-decreasing. This must be the case, since the optimal policy always assigns an arriving good to the highest-valued buyer, and a buyer who is removed is always the lowest-valued buyer in the queue.

\section{Proof of \cref{prop:thresholds-compstat}}\label{sec:appendix-thresholds-compstat}

If $\hat v_1=1$, then $K^*=0$, and the proposition is vacuously true. If $\hat v_1<1$,
it is clear from \eqref{eq33} that $\hat v_1$ is strictly increasing in $c$, constant in $\l$, and strictly decreasing in $\mu$.

For each $k \in\{2,\ldots, K^*\}$, the threshold $\hat v_k$ satisfies $H_k(\hat v_{k-1}, \hat v_k)=0$, which is equivalent to
\begin{align*}
    \m \int_{\hat v_{k-1}}^{\hat v_k} \frac{1}{1+\rho(v) + \cdots + \rho(v)^{k-1}} J'(v) dv = c,
\end{align*}
where we recall that $\rho(v) \:= \l[1-F(v)]/\m$.
Clearly, each $\hat v_k$ is strictly increasing in $c$, strictly increasing in $\l$, and strictly decreasing in $\mu$.

\section{Proof of \cref{thm:DSIC}}

We define the ``Truthful Strategy'' as behaving according to one's true valuation $v$ at every decision node. We show that deviating from this strategy is never profitable, regardless of past or future events.

\smallskip 
\noindent{\bf (i) Inventory Sales.} It is clearly weakly dominant to accept the good if and only if $v \ge \hat v_{-\ell}$.

\smallskip 
\noindent{\bf (ii) Survival Auctions.}
Suppose the current clock price is at $p$.
\begin{itemize}
    \item \textbf{Dropping out early ($p < v$):} Dropping out yields 0  payoff. Staying incurs no waiting cost (because it is reimbursed) and retains the option value of winning. 
    If he bids truthfully in the future ($b=v$), even if he wins, the cutoff price is lower than $v$, guaranteeing non-negative utility. 
    Thus, dropping out early is weakly dominated.

    \item \textbf{Staying late ($p > v$):} If the buyer stays past $v$, his personalized reserve $r_i$ updates to $p$. Even if he wins later, the cutoff price will be at least $r_i=p>v$, resulting in strictly negative utility. Thus, staying past $v$ is weakly dominated.
\end{itemize}

\noindent {\bf (iii) Assignment Auctions.}  Consider a buyer with value $v$ and reserve $r_i\le v$. (If $r_i> v$, the buyer would have optimally dropped out earlier). The buyer submits bid $b$. The mechanism sets cutoff price $P$ as the lowest bid $\mathcal B\ge r_i$ such that a Proxy Agent with bid $\mathcal B$ eventually wins on the realized path.  We examine deviations from truthful bidding ($b=v$): 

\smallskip
$\bullet$ \textbf{Overbidding ($b > v$):}  
The buyer's deviation may change the outcome of the assignment auction. However, the buyers who played the assignment auction become passive and are replaced by proxy agents (who stay and play according to their fixed bids), so their behavior cannot change in response to the outcome.
Future arriving buyers do not observe the deviation, as any information about past play or the current queue length is hidden. Hence, their strategies and the resulting future path of bids remain identical to the truthful scenario as long as the truthful bid would also have won. If the truthful bid would have lost, then overbidding results in a cutoff price above one's value.

    \begin{itemize}
        \item [-]\textit{Case A: The buyer wins (either in the current assignment auction, or in the future) with bid $v$.} Because opponents do not change their behavior, the cutoff price is identical, so the deviation is not profitable.
         
        \item [-]\textit{Case B: The buyer loses with bid $v$ but wins with bid $b$.}  Losing with $v$ implies that the cutoff price $P$ when one overbids and wins must exceed $v$.  Thus, winning with a bid $b>v$ yields a payoff of $v-P<0$, making deviation unprofitable.
        \item [-]\textit{Case C: The buyer loses with both $v$ and $b$.} The buyer's payoff is 0 in both cases.
    \end{itemize}

\smallskip 
$\bullet$  \textbf{Underbidding ($b\in [r_i, v)$):} The critical observation is that the cutoff price $P$ is invariant to the specific winning bid, provided it results in a win. Again, because remaining buyers become passive and future buyers do not observe the deviation, their behaviors do not change as long as the buyer still wins with $b$.
    \begin{itemize}

        \item [-] \textit{Case A: The buyer loses with $v$.}
        If the buyer loses with $v$, he also loses with a lower bid $b$. The payoff is 0 in both cases.
        \item [-] \textit{Case B: The buyer wins with $v$ but loses with $b$.}
        Winning with $v$ yields a surplus of $v - P \ge 0$. Losing with $b$ yields  0. The deviation is weakly worse.
        \item [-] \textit{Case C: The buyer wins with both $v$ and $b$.}  Since opponents do not change their behavior, the cutoff price is identical, so the deviation is unprofitable.
    \end{itemize}

 Combining (i), (ii), and (iii), we conclude that the Truthful Strategy is weakly dominant.

\section{Proofs for \cref{sec:large-market}}

\subsection{Proof of \cref{lemma:large-market-1}}

Fix any policy $\phi \in \Phi$ on any $(m,n)$-sample path. The policy $\phi$ induces an ex-post allocation rule, specifying the allocation as a function of realized values for the $n$ buyers. The allocation rule is monotonic, making it feasible in the Oracle benchmark. Since the Oracle benchmark involves no waiting costs for buyers or holding costs for items, $R^*(m,n)\ge R_\phi(m,n)$.

\subsection{Proof of \cref{lemma:large-market-2}}

    Whenever $\l,\m\to \infty$ with $\rho=\l/\m$ held constant, the numbers of buyers and items arriving during a unit interval, $n$ and $m$, must go to $\infty$ with $n/m\to \rho$.  Hence, the result follows: 
\begin{align*}
    & \limsup_{\l,\m\to\infty, \l/\m=\rho} \sup_{\phi \in \Phi}\E_{n\sim \l,m\sim \m}[ R_\phi(m,n)/m ]\\
    \le &  \limsup_{n,m\to\infty, n/m=\rho} \sup_{\phi \in \Phi} R_\phi(m,n)/m 
 \\
 \le &  \lim_{n,m\to\infty, n/m=\rho} R^*(m,n)/m \\
    = & R^*.
\end{align*}

\subsection{Proof of \cref{thm:large-market}}

By the arguments following the theorem, (a) follows from either (b) or (c).

\paragraph*{Proof of (b).}

We consider the case of $\l,\m,d \to \infty$ with $\rho = \l/\m$ and $\l/d$ held constant; as previously discussed, this is equivalent to $c\to 0$ up to rescaling of time. The upper bound of $R^*$ obtained in \cref{lemma:large-market-2} still applies, since also sending $d \to \infty$ cannot increase revenue.
It thus suffices to show that there exists a policy $\phi \in \Phi$ whose expected per-unit revenue converges to $R^*$. Consider the following simple policy:  For any  $\epsilon>0$, admit a buyer if and only if his valuation is above
$$\tilde v_{\delta} \:= \max\{\tilde v, \hat  v_0\}+\delta,$$
where $\delta>0$ is chosen so that $R^*-\rho\tilde v_{\delta}[1-F(\tilde v_{\delta})]<\epsilon$.  
If a good arrives, it is discarded if there are no buyers in the queue; otherwise,  it is allocated to a randomly chosen buyer in the queue at price  $\tilde v_{\delta}$; the seller reimburses buyers for their waiting costs.  Under this policy, all buyers with valuations above $\tilde v_{\delta}$ are eventually allocated the good.  

The expected per-unit  revenue for the seller is then
$$\frac{\l}{\m}\tilde v_{\delta}[1-F(\tilde v_{\delta})]$$
minus the per-unit expected waiting cost, or the total waiting cost divided by the number of items.  We show that the latter vanishes in the limit.

\begin{lem}  The per-unit expected waiting cost vanishes in the limit as $\l, \m\to \infty$ with $\rho=\l/\m$ held constant. 
\end{lem}

\begin{proof} Since $c$ is held fixed, it is enough to show that the per-unit expected waiting time vanishes. Observe that $\rho_{\delta}:=\rho [1-F(\tilde v_{\delta})]<1.$ 
Let $p_k$ denote the stationary probability that the queue has length $k=0, ..., \infty$. For each $k=0,...$,
we must have 
$\lambda [1-F(\tilde v_{\delta})] p_k= \mu p_{k+1},$
so 
$p_{k}=\rho^{k}_{\delta}/\sum_{k=0}^{\infty} \rho^{k}_{\delta}.$
Since $\rho_{\delta} <1$,
the average length of the queue 
\begin{eqnarray*}
\sum_{k=0}^{\infty} kp_{k} &=&\sum_{k=0}^{\infty}k  \frac{\rho^{k}_{\delta}}{%
\sum_{k=0}^{\infty} \rho^{k}_{\delta}} = \frac{\rho_{\delta}}{1-\rho_{\delta}}
\end{eqnarray*}%
is finite.  Since the effective arrival rate 
$\lambda [1-F(\tilde v_{\delta})]$
tends to infinity, by Little's law, the per-capita expected waiting time vanishes as $\l\to \infty$.  This means the per-unit expected waiting time (obtained by multiplying $\rho_{\delta}$) also vanishes in the limit. 
\end{proof}

Combining the preceding observations, we conclude that the limit expected (per-unit) revenue from the simple policy described above is no less than  
$R^*-\epsilon$. Therefore, the same conclusion applies to our optimal policy.

\paragraph*{Proof of (c).}

Analogously to the proof of (b), we will consider the case $\l,\m,c \to \infty$ with $\rho = \l/\m$ and $\l/c$ held constant. Our goal is to show that there exists a policy $\phi\in\Phi$ whose expected per-unit revenue converges to $R^*$. Consider the following simple policy.  If a buyer arrives, he is turned away if there are no items stored; otherwise, for $\epsilon>0$, the buyer is allocated an item at a posted price of
$\tilde v_{\delta} =\max\{\tilde v, \hat  v_0\}+\delta,$
where $\delta>0$ is chosen so that $R^*-\rho\tilde v_{\delta}[1-F(\tilde v_{\delta})]<\epsilon$. 
Arriving goods are stored up to a cap of at most $\bar L \:= \left[ \sqrt{\mu} \right] $ in inventory.
Under this policy, the expected per-unit storage time is at most $\frac{\sqrt{\mu}}{\mu}$, which tends to 0 as $\mu \to \infty$.

Let $q_\ell$ denote the stationary probability under this policy that there are $\ell$ items in storage. For each $\ell  = 0,\ldots,\bar L-1$, we must have $\mu q_\ell = \l[1-F(v_\d)]q_{\ell+1}.$
Therefore, the probability that the storage is empty is
$$q_0 = \frac{1}{1+\psi_\d + \cdots + \psi_{\d}^{\bar L}},$$
where $ \psi_{\d} \:= 1/\rho[1-F(\tilde v_\d)] = 1/\rho_{\d}.$
Since $\psi_\d >1$, it must be that $q_0$ tends to 0 as $\bar L \to \infty$.
Therefore, buyers with valuations above $\tilde v_{\delta}$ are always allocated a good in the limit.

Combining the preceding observations, we conclude that the limit expected (per-unit) revenue from the simple policy desribed above is 
$$\frac{\l}{\m}\tilde v_{\delta}[1-F(\tilde v_{\delta})] > R^* - \epsilon.$$
Therefore, the same conclusion applies to our optimal policy.

\section{Proof of \cref{prop:comp-stat-in-pareto-weight}} \label{sec:appendix-proof-of-comp-stat-in-pareto-weight}

It is enough to show that if $w > w'$ and $\hat v_k^w > \hat v_k^{w'}$ for some $k \leq \min\{K_{w'}, K_w\}-1$, then $\hat v_{k+1}^w > \hat v_{k+1}^{w'}$. 
Assume for the sake of contradiction that $\hat v_k^w > \hat v_k^{w'}$ and $\hat v_{k+1}^w \leq \hat v_{k+1}^{w'}$ for some $k$. The thresholds must satisfy%
\begin{align}
    \m \int_{\hat v_{k}^w}^{\hat v^w_{k+1}} \frac{1}{1+\rho(v) + \cdots + \rho(v)^{k}} J'_w(v) dv = c \label{eq52} \\
    \m \int_{\hat v_{k}^{w'}}^{\hat v^{w'}_{k+1}} \frac{1}{1+\rho(v) + \cdots + \rho(v)^{k}} J'_{w'}(v) dv = c.\label{eq53}
\end{align}
Because the inverse hazard rate is non-increasing, it must be that $J'_{w'}(\cdot) \geq J'_w(\cdot) \geq 1$. Since we assumed that $\hat v_k^w > \hat v_k^{w'}$ and $\hat v_{k+1}^w \leq \hat v_{k+1}^{w'}$, the LHS of \eqref{eq53} must be strictly greater than the LHS of \eqref{eq52}, a contradiction.

\newpage

\section*{\centering  \fontsize{20pt}{24pt}\selectfont Online Appendix:  ``Not For Publication''}
\addcontentsline{toc}{section}{Online Appendix}
\label[onlineapp]{sec:online-appendix}

\bigskip
\bigskip

\section{Markov Policies are Without Loss}\label[onlineapp]{sec:appendix-markov-wlog}

Here, we relax the Markov restriction and allow the seller to adopt policies which may depend generally on the entire history of buyer and item arrivals. We argue that the optimal policy we derived in \cref{sec:optimal-mechanism} is in fact optimal within this broader class of policies. 

\paragraph*{Seller's Dynamic Policy for $d=\infty$.}\label{subsec:seller's-dynamic-policy}

Consider a probability space $(\Om,\ms F, \Prob)$, where each sample $\om \in \Om$ consists of the {\it arrival times} $T_i^B$ of buyers, indexed by their arrival orders $i\in \mathbb{N}$, and {\it their values} $V_i$, and  the {\it arrival times} $T^S_j$ of all items, again indexed by their arrival orders $j\in \mathbb{N}$. To allow for a stochastic policy, we let $\om$   contain a sequence  $\{\xi_i\}_{i=1}^\infty$ of {\it random numbers} each drawn uniformly from $[0,1]$ at the time of buyer $i$'s arrival.

For each sample $\om \in \Om$, a \textbf{policy} $\phi(\om) = (a(\om),r(\om),\tau(\om))$ specifies three mappings:\footnote{Throughout this section, we redefine notation used in the main text. For example, here we are redefining $\phi,a,r,$ and $\tau$.}

\begin{itemize}
    \item $a(\om):\N \to \N \cup \{0\}$, where $a(\om)(i) = j$ means that the $i$-th arriving buyer is \textit{assigned} the $j$-th arriving item (and $a(\om)(i) = 0$ means he leaves the queue without being assigned),
    \item $r(\om):\N \to \R_+$, where $r(\om)(i)$ is the time at which the $i$-th buyer is \textit{removed} from the queue,
    \item $\tau(\om):\N \to \R$ is the total amount of \textit{transfer} made by the $i$-th buyer. We assume that the entire $\tau(\om)(i)$ is paid at time $r(\om)(i)$, when the buyer leaves.
\end{itemize}
To be feasible, a policy must satisfy the following conditions for each $\om$:\footnote{We drop the argument $\om$ when unambiguous.} (i) if $a(i) = a(i')$ for $i\neq i'$, then $a(i) = a(i') = 0$ (no two buyers can receive the same item); (ii) $r(i) \geq T_i^B$, where $T_i^B$ denotes the arrival time of the $i$-th arriving buyer (buyers cannot leave before arriving);  (iii) if $a(i) = j \neq 0$, then $r(i) = T_j^S$, where $T_j^S$ denotes the arrival time of the $j$-th arriving item (items are assigned immediately or discarded, and assigned buyers leave immediately); (iv) $a$, $r$, and $\tau$ are adapted to the natural filtration with respect to the arrival processes of buyers and items, as well as the randomization devices  (policy cannot condition on future events).

A policy as defined here is allowed to depend arbitrarily on the entire history of buyer and item arrivals. Also, the policy may remove a buyer from the queue at any time, in any manner. For example, the policy could remove a buyer immediately upon arrival (meaning the buyer never joins the queue); remove all buyers with values above $v$ with some probability whenever an item is assigned or when a new buyer arrives; or remove the lowest-value buyer at some constant Poisson rate, independently of the arrival of buyers or items.  

\paragraph*{Regenerative Policies.}

We wish to define a notion of steady state and evaluate the seller's expected revenue under this steady state.
Toward this goal, first define 
$$B_t(\om,\phi)\:= \{ i\in \N \mid T_B^i \leq t \text{ and } r(i) > t\}$$
to be the set of buyers who are in the queue at time $t$, given sample $\om$ and policy $\phi$. A time $T^0> 0$ is \textbf{null} given $\om$ and $\phi$ if (i) $B_{T^0}(\om,\phi) = \emptyset$, and (ii) $B_s(\om,\phi) \neq \emptyset$ for all $s \in [T^0-\varepsilon,T^0)$, for some $\varepsilon>0$. That is, a null time is a moment at which the queue of buyers becomes empty.
For convenience, time 0 is defined to be null for any $(\om, \phi)$.
When necessary, we write $T^0_k(\om,\phi)$ to make explicit the dependency on sample $\om$ and policy $\phi$.

Let $\{T^0_k\}_{k\in\N}$ denote the random sequence of null times, with $T^0_1=0$. The interval $[T^0_k,T^0_{k+1})$ is called the $k$-th \textbf{cycle}.

When a new cycle begins, it is as if the arrival processes of buyers and items start again from time 0. 
Since the arrival processes of buyers and items are Markovian, any event in past or future cycles is irrelevant for the seller's revenue in the current cycle, and also irrelevant for the payoffs (and hence incentives) of buyers who are in the current cycle. 
Therefore, it is  without loss for the seller to separately maximize her expected revenue under each cycle; moreover, it is without loss for the seller to repeat the same policy at every cycle. 

To formalize this idea, let $\om \vert_t$ denote the truncation of $\om$ from time $t$ onward, with time reset to 0 and the buyer and item indices reset to 1. That is, $$ \om \vert_t \:= \left(\left\{T_{I_t+i}^B - t,V_{I_t+i}\right\}_{i\in\N},\left\{T_{J_t+j}^S - t\right\}_{j\in\N}, \left\{\xi_{s+I_t}\right\}_{s=1}^\infty\right), $$
where $I_t$ and $J_t$ denote the number of buyers and items, respectively, that have arrived up to time $t$.  Similarly, let 
$\phi|_t$ denote the truncation of $\phi$ from $t$ onward, again with the time reset to 0 and the buyer and item indices reset to 1. That is, $\phi|_t: = (a|_t,r|_t,\tau|_t)$, where
\begin{align*}
a|_t(\om)(i) &\, \:= \max\{a(\om)(i + I_t) - J_t, 0 \},\\
r|_t(\om)(i) &\, \:=  r(\om)(i + I_{t}) - t,\\
\tau|_t(\om)(i) &\, \:= \tau(\om)(i + I_{t}),
\end{align*} for each $\om\in\Om$ and $i\in \N$.  A policy $\phi$ is \textbf{regenerative} if, for each $\om\in\Om$,  
\begin{align}\label{eq:regenerative_policy}
    \phi|_t(\om) = \phi(\om|_t) \qquad \forall t = T^0_1,T^0_2,\ldots.
\end{align}
That is, under a regenerative policy, the outcome process following each null time (the LHS of \eqref{eq:regenerative_policy}) is the same as that starting from time 0 (the RHS of \eqref{eq:regenerative_policy}).
Whenever the queue empties, a regenerative policy ``forgets'' the past and continues as if the system restarted from time 0.

\paragraph*{Steady State Analysis.}  

For a given $\om$ and a regenerative policy $\phi$, we focus on  \textbf{within-cycle history at time $t \in [T^0_k, T^0_{k+1})$}, defined as 
$$\t^\phi_t(\om) \:=  \left(\left\{T_{I_{T^0_k}+i}^B - T^0_k,V_{I_{T^0_k}+i}\right\}_{i=1}^{I_t - I_{T^0_k}},\left\{T_{J_{T^0_k}+j}^S - T^0_k\right\}_{j=1}^{J_t - J_{T^0_k}}, \left\{\xi_{I_{T^0_k}+s}\right\}_{s=1}^{I_t - I_{T_k^0}}, t- T^0_k \right) .$$ 
That is, $\t^\phi_t(\om)$ records all events that occurred after $T^0_k$ and up to $t$, with both time and indices reset.

Let $\Theta$ be the set of all within-cycle histories that may be obtained under some $\om$, $\phi$, and $t$. A regenerative policy induces a stochastic process $\{\t_t\}_{t\geq 0}$ which takes the set $\Theta$ as the state space. Importantly, since sequence of null times constitutes a renewal process (this is because the amount of time spent between two adjacent null times is iid), $\{\t_t\}_{t\geq 0}$ is a {\bf regenerative process} which repeats the same process, or ``regenerates,'' following each null time.\footnote{For a formal definition of renewal or regenerative processes, see \cite{asmussen2003applied}.}
We restrict attention to a regenerative policy $\phi$ whose induced process  $\{\t_t\}_{t\geq 0}$ admits a stationary distribution $\pi \in \D(\Theta)$.

Let $\Phi$ denote the set of all incentive compatible regenerative policies. Given a policy $\phi = (a,r,\tau) \in \Phi$, we abuse notation and define 
$$\tau(v;\t) \:= \E\left[ \tau(\om)(i) \,\middle|\, \t^{\phi}_{T^{B_i-}}(\om) = \t ,\,\,V_i = v \right],$$ 
where $\t_{T^B_i-}$ denotes the within-cycle history immediately before the $i$-th buyer arrives.  That is, $\tau(v;\t)$ denotes the expected payment of a buyer with value $v$ who arrives at within-cycle history $\t \in \Theta$.
Note that the within-cycle history $\t$ determines $i$ only up to the buyer's ``within-cycle'' order, but nevertheless $\tau(v;\t)$ is well-defined because $\phi$ is regenerative. 

The \textbf{seller's problem} is
\begin{align*}
 [\mathcal P_0] \hspace{9em} \sup_{\substack{\phi\in  \Phi \\ \pi \in \D (\Theta)}}  &\, \int_{\t}\int_0^1 \l \tau(v; \t) f(v)\ddd v\,\pi(\dd \t)  \hspace{10em} \\
\text{subject to } &\, \phi \text{ inducing $\pi$.}
\end{align*}
The relaxed program described in \cref{subsec:relaxed} is clearly a relaxation of this generalized version of the program $[\ms P_0]$. Moreover, since the optimal Markovian policy described in \cref{sec:optimal-mechanism} is positive recurrent, it induces a stochastic process $\{\t_t\}_{t\geq 0}$ with a unique stationary distribution. Hence, we may repeat the argument in \cref{sec:appendix-proof-of-thm:optimal-policy} to show that this unique stationary distribution has projection $P^*$. Therefore, the optimal Markovian policy is also a solution to the general version of $[\ms P_0]$ formulated here. In this sense, the restriction to Markov policies is without loss.

\section{Omitted Arguments in the Proof of \cref{thm:relaxed-program,thm:relaxed-program-storable}}

\subsection{Constructing $\beta$ given $\hat v_k$'s}\label[onlineapp]{sec:appendix-multipliers}

We show how to construct the multiplier $\b$ so that the coefficient of $P_k(v) - P_k(0)$ in \eqref{eq20} equals 0 for $v \geq \hat v_k$.
Replacing $\eta_k(v)$ in using \eqref{eq23}, the coefficient of $P_k(v) - P_k(0)$ for $k\geq 2$ becomes\footnote{We defined $\eta_0(v)=0$, so \eqref{eq23} does not hold for $k=0$, which means the case of $k=1$ must be dealt with separately.}
\begin{align}
    \l[1-F(v)](\beta_k(v) - \beta_{k+1}(v)) - \mu(\beta_{k-1}(v) - \beta_k(v)), \label{eq21}
\end{align}
and for $k = 1$, the coefficient becomes
\begin{align}
    \l[1-F(v)](\beta_1(v) - \beta_2(v)) - \mu\alpha(v) + \mu\beta_1(v). \label{eq25}
\end{align}
Define $h_k(v)\:=\b_k(v) - \b_{k-1}(v)$, and recall $\sigma(v)\:= 1/\rho(v) = \mu/\l[1-F(v)]$. For the coefficients \eqref{eq21} and \eqref{eq25} to vanish, we want
\begin{align}
    & h_{k+1}(v) = \sigma(v) h_k(v) \qquad &\forall k\geq 2, \forall v \geq \hat v_k \label{eq:2-sided-1} \\
    & h_2(v) = \sigma(v) \left(\b_1(v) - J'(v) \right) \qquad &\forall v\geq \hat v_1. \label{eq:2-sided-3}
\end{align}
Use \eqref{eq:2-sided-1} to write
\begin{align*}
    \b_k(v) - \b_{k-1}(v) = h_{k}(v) = \sigma(v)^{k-2}h_2(v) \qquad \forall k\geq 2,\, \forall v \geq \hat v_{k-1}.
\end{align*}
Summing across $k$, we get
\begin{align}
    \beta_k(v) = \beta_1(v) + h_2(v) (1 + \sigma(v) + \cdots + \sigma(v)^{k-2}) \qquad \forall k\geq 2,\,\forall v\geq\hat v_{k-1}. \label{eq:2-sided-2}
\end{align}

It remains to pin down $\beta_1$ and $h_2$, or equivalently, $\beta_1$ and $\beta_2$. We do so using two boundary conditions. The first condition is \eqref{eq:2-sided-3}. To obtain the second condition,
fix an integer $i\geq 1$, and consider the case $v \in [\hat v_{i},\hat v_{i+1})$.
Since we set $\beta_k(v)=0$ for all $v<\hat{v}_k$, \eqref{eq:2-sided-2} implies that
\begin{align}
    \b_{i+1}(v) =  0 = \beta_1(v) + h_2(v) (1 + \sigma(v) + \cdots + \sigma(v)^{i-1}). \label{eq:2-sided-4}
\end{align}
Using \eqref{eq:2-sided-3} and \eqref{eq:2-sided-4}, we obtain
\begin{align*}
    &\beta_1(v) = \frac{\sigma(v) + \sigma(v)^2 + \cdots + \sigma(v)^i}{1 + \sigma(v) + \cdots + \sigma(v)^{i}} J'(v)\\
    &\beta_2(v) = \frac{\sigma(v)^2 + \sigma(v)^3 + \cdots + \sigma(v)^i}{1 + \sigma(v) + \cdots + \sigma(v)^{i}}J'(v).
\end{align*}
Therefore, using \eqref{eq:2-sided-2}, we obtain the closed form solution for $\beta$ as stated in \cref{subsec:multipliers}.

\subsection{Existence of Thresholds}\label[onlineapp]{subsec:IVT}

For exposition, redefine 
\begin{align}
A(\g_1, \g_2) :=\d(\g_1) + \int_{J^{-1}(\g_1)}^1 \l(J(v) - \g_1) \ddd F(v) - \mu(\g_1 - \g_2) -d.  \notag
\end{align}
where $\d(\g_1):= -\int_0^1 \l[1-F(v)]\beta_1(v;J^{-1}(\g_{1}))dv$, and
\begin{align*}
   B(\g_{\ell-1}, \g_\ell, \g_{\ell+1}):=    &\,\l(\g_{\ell-1} - \g_\ell)(1-F(J^{-1}(\g_{\ell-1}))) \\
   &\,+ \int_{J^{-1}(\g_{\ell})}^{J^{-1}(\g_{\ell-1})} \l(J(v)-\g_\ell) \ddd F(v) - \mu(\g_\ell - \g_{\ell+1}) - d.
\end{align*}
$A(\g_1, \g_2)$ is increasing in $\g_2$. 
$B(\g_{\ell-1}, \g_\ell, \g_{\ell+1})$ is increasing in $\g_{\ell-1}$ and $\g_{\ell+1}$, and is decreasing in $\g_\ell$. Because $J(v)$ is continuous, and because the buyer thresholds $\hat v_k$ vary continuously as one varies $\hat v_{-1}$, $A$ and $B$ are also continuous. Moreover, $A(1,\g_2)<0$ for all $\g_2$, and $B(\g_{\ell-1},\g_\ell,\g_{\ell+1})<0$ whenever $\g_{\ell-1} = \g_\ell$.

If $A(0,0) \leq 0$, then we set $L^*=0$, and we are done. Otherwise, since $A(0,0)>0$ and $A(1,1)<0$, there  exists $\gh^1 \in (0,1)$ such that $A(\gh^1, \gh^1)=0.$  
Without loss, assume that $\gh^1$ is the smallest such value. In other words, $A(\g,\g)>0$ for all $\g\in[0,\gh^1)$. Next, $A(\gh^1, \g)<0$ for all $\g<\gh^1$.  Hence, there exists $\g^0<\gh^1$ such that $A(\g^0,0)=0$.  Without loss, take $\g^0$ to be the largest such value, namely:  $A(\g,0)<0$ for all $\g\in (\g^0, \gh^1]$. If $B(\g^0,0,0) \leq 0$, then we choose $\g_1 = \g^0$ and $L^*=1$, and we are done.\\

Next, suppose $B(\g^0, 0,0)>0 $.
For any $\g\in [\g^0,\gh^1]$, since $A(\g,0)\leq 0$ and $A(\g,\g) \geq 0$, there exists $\nu_2(\g)\in [0,\gh^1]$ such that 
\begin{align}
    A(\g,\nu_2(\g))=0. \notag
\end{align}
Observe that $\nu_2(\gh^1)=\gh^1$ and 
$\nu_2(\g^0)=0$.  Further, since $A$ is continuous and is strictly increasing in its second argument, $\nu_2(\cdot)$ is continuous.  

Then, 
$B(\g^0,\nu_2(\g^0),0)>0$ and since 
$B(\gh^1,\nu_2(\gh^1), 0)<0$, there exists 
$\g^1\in [\g^0, \gh^1]$ such that 
 \begin{align} \label{eq:B}
    B(\g^1,\nu_2(\g^1), 0)=0.
\end{align}
Without loss, take $\g^1$ to be the largest such value. If $B(\nu_2(\g^1),0,0) \leq 0$, then we may choose $\g_1 = \g^1$, $\g_2 = \nu_2(\g^1)$, and $L^*=2$, and we are done.\footnote{Note that $\g_2>0$ because $B(\g^1,0,0) \geq B(\g^0,0,0) > 0$.} \\

At this point, we may proceed directly to the general induction argument, but for expositional purposes, we first work through one additional step.
Suppose that 
$$B(\nu_2(\g^1),0, 0)>0. $$

Since $B$ is increasing in its last argument, \eqref{eq:B} implies that $B(\g^1,\nu_2(\g^1),\nu_2(\g^1))>0$. Also, since $\nu_2(\gh^1) = \gh^1$, it must be that
$B(\gh^1,\nu_2(\gh^1),\nu_2(\gh^1))<0$. Thus there exists $\gh^2\in (\g^1 , \gh^1)$ such that 
$$B(\gh^2,\nu_2(\gh^2),\nu_2(\gh^2))=0.$$  Without loss, let $\gh^2$ be the smallest such value.

Fix any $\g\in [\g^1, \gh^2]$.  Then, 
$$B(\g,\nu_2(\g),\nu_2(\g))\ge 0.$$
This follows from the definition of $\gh^2$.
Note also that 
$$B(\g,\nu_2(\g),0)\le 0.$$
This follows from the definition of $\g^1$.

Hence, there exists $\nu_3(\g)\in [0,\nu_2(\g)]$ such that 
 \begin{align}
    B(\g,\nu_2(\g), \nu_3(\g))=0. \notag
\end{align}
Because $B$ is strictly increasing in its third argument, it must be that 
$\nu_3(\g^1)=0$ and $\nu_3(\gh^2)=\nu_2(\gh^2)$.
Moreover, because $B$ is continuous, and is strictly increasing in its third argument, $\nu_3$ is continuous.

Recall our hypothesis:   
$$B(\nu_2(\g^1),\nu_3(\g^1),0)=B(\nu_2(\g^1),0, 0)>0. $$
Since  $\nu_3(\gh^2)=\nu_2(\gh^2)$, we have 
$$B(\nu_2(\gh^2),\nu_3(\gh^2),0)<0.$$
Hence, there exists $\g^2\in [\g^1,\gh^2]$ such that 
 \begin{align} 
    B(\nu_2(\g^2),\nu_3(\g^2),0)=0.\notag
\end{align}
Without loss, let $\g^2$ be the largest such value. If $B(\nu_3(\g^2),0,0) \leq 0$, then we choose $\g_1 = \g^2$, $\g_2 = \nu_2(\g^2)$, and $\g_3 = \nu_3(\g^2)$, and choose $L^*$ to be the largest $\ell\in \{1,2,3\}$ such that $\g_\ell >0$, and we are done.

We now make the general induction argument.

\begin{defn}[Induction hypothesis at $k \geq 2$]
There exist $\g^{k-3}, \g^{k-2}, \g^{k-1}, \gh^{k-1}, \gh^{k-2} \in [0,1]$ with $\g^{k-3} \leq \g^{k-2} \leq \g^{k-1} \leq \gh^{k-1} \leq \gh^{k-2}$, a continuous function $\nu_{k-1}:[\g^{k-3},\gh^{k-2}] \to [0,\gh^{k-2}]$, and a continuous function $\nu_k:[\g^{k-2},\gh^{k-1}] \to [0,\gh^{k-1}]$ that satisfy the following conditions.\footnote{When $k=2$, define $\g^{-1}:=1$ and $\gh^0:=1$, and define $\nu_1$ to be the identity function on $[0,1]$.}
\begin{enumerate}%
    \item $B(\nu_{k-1}(\g^{k-1}),\nu_k(\g^{k-1}),0) = 0$.
    \item $B(\nu_{k-1}(\g),\nu_k(\g),0) < 0$ for all $\g \in (\g^{k-1},\gh^{k-1}]$.
    \item $\nu_k(\g^{k-2})=0$ and $\nu_k(\gh^{k-1}) = \nu_{k-1}(\gh^{k-1})$.
    \item $\nu_k(\g) \leq \nu_{k-1} (\g)$ for all $\g \in [\g^{k-2},\gh^{k-1}]$.
\end{enumerate}
\end{defn}
We showed above that either a solution exists with $L^* \leq 2$, or the induction hypothesis holds at $k=2$. We wish to show that, for $k\geq 2$, if
$$B(\nu_k(\g^{k-1}),0, 0)>0, $$
then the induction hypothesis at $k$ implies the induction hypothesis at $k+1$.
Since $B$ is strictly increasing in its last argument, induction hypothesis (a) implies that 
$$B(\nu_{k-1}(\g^{k-1}),\nu_k(\g^{k-1}),\nu_k(\g^{k-1}))>0.$$ 
Also, since $\nu_k(\gh^{k-1}) = \nu_{k-1}(\gh^{k-1})$ by hypothesis (c), it must be that
$$B(\nu_{k-1}(\gh^{k-1}),\nu_k(\gh^{k-1}),\nu_k(\gh^{k-1}))<0.$$ 
Thus there exists $\gh^k\in (\g^{k-1} , \gh^{k-1})$ such that 
\begin{align}\label{eq:ghk}
    B(\nu_{k-1}(\gh^k),\nu_k(\gh^k),\nu_k(\gh^k))=0.
\end{align}
Without loss, let $\gh^k$ be the smallest such value.

Fix any $\g\in [\g^{k-1}, \gh^k]$.  Then, by our choice of $\gh^k$, it must be that
$$B(\nu_{k-1}(\g),\nu_k(\g),\nu_k(\g))\ge 0.$$
Moreover, by induction hypothesis (b), we have
$$B(\nu_{k-1}(\g),\nu_k(\g),0)\le 0.$$
Hence, there exists $\nu_{k+1}(\g)\in [0,\nu_k(\g)]$ such that 
 \begin{align*} 
    B(\nu_{k-1}(\g),\nu_k(\g), \nu_{k+1}(\g))=0.
\end{align*}
Thus hypothesis (d) holds for $k+1$.
Because $B$ is strictly increasing in its third argument, it must be that 
$\nu_{k+1}(\g^{k-1})=0$ (by induction hypothesis (a)) and $\nu_{k+1}(\gh^k)=\nu_k(\gh^k)$ (by \eqref{eq:ghk}). Thus hypothesis (c) is satisfied for $k+1$.
Moreover, because $B$ is continuous, and is strictly increasing in its third argument, $\nu_{k+1}$ is continuous.

Recall that we assumed   
$$B(\nu_k(\g^{k-1}),\nu_{k+1}(\g^{k-1}),0)=B(\nu_k(\g^{k-1}),0, 0)>0. $$
Since  $\nu_{k+1}(\gh^k)=\nu_k(\gh^k)$, we have 
$$B(\nu_k(\gh^k),\nu_{k+1}(\gh^k),0)<0.$$
Hence, there exists $\g^k\in [\g^{k-1},\gh^k]$ such that 
\begin{align*} %
    B(\nu_k(\g^k),\nu_{k+1}(\g^k),0)=0.
\end{align*}
Without loss, let $\g^k$ be the largest such value. This satisfies hypotheses (a) and (b) for $k+1$.

Finally, for large enough $k$, we must have $B(\nu_k(\g^{k-1}),0,0)\leq 0$.
This is because the sum
\begin{align*}
    &\, A(\g_1, \g_2) + \sum_{\ell=2}^{\hat{L}} B(\g_{\ell-1}, \g_\ell, \g_{\ell+1}) \\
    = &\, \d(\g_1) + \int_{J^{-1}(\g_{\hat L})}^1 \l(J(v) - \g_{\hat L}) dF(v) - \mu(\g_1 - \g_{\hat L +1}) - \hat L d
\end{align*}
eventually becomes negative when $\hat L$ is large enough. When $B(\nu_k(\g^{k-1}),0,0)\leq 0$, there exists a solution with $L^* \leq k$.

\subsection{Finiteness of $K$ and $L$} \label[onlineapp]{subsec:K=L=infty}

We prove that it is without loss of optimality to restrict attention to policies with bounded queue lengths.

\begin{prop}
    There exists an optimal policy such that, at any point in time, at most $K \in \N$ buyers are in the queue, and at most $L \in \N$ goods are stored.
\end{prop}
\begin{proof}
We will argue that, even if we allow the queue lengths to be unbounded ($K=L=\infty$), it will be optimal for the seller to keep the queue lengths finite. It is enough to prove this claim for the relaxed program.

Take $K$ to be a large enough integer such that $\max\{K^*,L^*\} < K$, and the following inequality holds: $$(K+1)d > \int_0^1 \max\{\l J(v),0\} f(v) dv.$$
Let $\ms L(K,L)$ denote the Lagrangian of the relaxed program with bounds $K$ and $L$. In defining $\ms L(\infty,\infty)$, we impose the following additional constraints:

\begin{itemize}
    \item $p_k(0)=P_{k+1}(0) - P_k(0)\ge 0$ for each $k\ge K+1$.
    \item $\sum_{k=1}^{\infty} p_k(0)=1$
    \item $Q_\ell$ is non-increasing in $\ell$ for $\ell \geq K+1$.
\end{itemize}
We may freely impose these constraints on the relaxed program $\ms L(\infty,\infty)$, since any feasible solution to the original program $[\ms P_0]$ must satisfy them.\footnote{Recall $p_k(0)$ is the probability that there are exactly $k$ buyers in the queue. This probably must be nonnegative and add up to one, to be feasible under $[\ms P_0]$.}

With these (free) constraints, we can write:
\begin{align}\label{eq:H3.1}
    \sum_{k=K+1}^\infty k (P_{k+1}(0) - P_k(0)) & =\sum_{k=K+1}^\infty k p_k(0)\cr
    &=(K+1)\sum_{k=K+1}^\infty p_k(0) +\sum_{k=K+1}^\infty (k-(K+1)) p_k(0)  \cr
    &=(K+1)(1-P_{K+1}(0))+ \sum_{k=K+2}^\infty (k-(K+1)) p_k(0)\cr
   & = (K+1)(1-P_{K+1}(0))+\sum_{k=K+2}^\infty \sum_{j=K+2}^\infty \mathbb{I}_{\{j\le k\}} p_k(0)  \cr
     &= (K+1) (1- P_{K+1}(0)) +\sum_{j=K+2}^\infty \sum_{k=K+2}^\infty  \mathbb{I}_{\{j\le k\}} p_k(0)\cr
         &= (K+1) (1- P_{K+1}(0)) +\sum_{j=K+2}^\infty \sum_{k=j}^\infty p_k(0)\cr
&= (K+1) (1- P_{K+1}(0)) + \sum_{j=K+2}^\infty (1-P_j(0)),
\end{align}
where the exchange of summations (fifth equality) follows from Tonelli's Theorem, noting that each summand $p_k(0)$ is nonnegative.

By factoring out $\ms L(K,K)$ from $\ms L(\infty,\infty)$, and using \eqref{eq:H3.1}, we obtain\footnote{For any feasible choice of primal variables, the infinite sums are well-defined in the space of extended reals, since the summands are nonnegative.}
\begin{align*}
     \ms L (\infty,\infty)  = &\, \ms L(K,K) - \int_0^1 \eta_K(v) dv + \int_0^1 \eta_K(v)P_{K+1}(v) dv  \\
    &\, - c \sum_{k=K+1}^\infty [1-P_k(0)]  - d \sum_{\ell=K+1}^\infty \ell(Q_\ell - Q_{\ell+1}) \\
    &\, +  \int_0^1 \alpha (v) \left(\mu
    \sum_{k=K+1}^{\infty} Y_k(v) + \l \sum_{\ell=K+1}^\infty \int_v^1 Z_\ell(s)ds \right) dv \\
    &\, + \sum_{k=K+1}^\infty \int_0^1 \beta_k(v) \bigg( \l (P_k(v) - P_{k-1}(v))[1-F(v)] - \mu Y_k(v) \bigg) dv \\
    &\, + \sum_{\ell=K+1}^\infty  \gamma_\ell \bigg( - \l \int_0^1 Z_\ell(s)ds + \mu (Q_{\ell-1} - Q_\ell) \bigg) \\
    &\, + \sum_{k=K+1}^\infty \int_0^1 \eta_k(v) \left( P_{k+1}(v) - P_k(v) - Y_k(v) \right) dv \\
    &\, + \sum_{\ell=K+1}^\infty \int_0^1 \kappa_\ell(v)\left( (Q_\ell - Q_{\ell+1})f(v) - Z_\ell(v) \right) dv.
\end{align*}
We set the multipliers as follows:
\begin{align*}
\beta_k(v) &= 0 && \forall k \geq K+1,\\
\eta_k(v) &= \mu \alpha(v) && \forall k \geq K+1,\\
\gamma_\ell &= 0 && \forall \ell \geq K+1,\\
\kappa_\ell(v) &= \lambda \int_0^v \alpha(s)\,ds && \forall \ell \geq K+1.
\end{align*}
We choose the following primal variables:
\begin{align*}
    P_k(v) &= 1 && \forall k \geq K+1, \, \forall v \in [0,1],\\
    Y_k(v) &=0  && \forall k \geq K+1, \, \forall v \in [0,1],\\
    Z_\ell(v) &= 0 && \forall \ell \geq K+1, \, \forall v \in [0,1],\\
    Q_\ell &= 0 && \forall \ell \geq K+1.
\end{align*}
Given our choice of $K$, one can verify that the conditions for weak duality are satisfied---the Lagrangian is maximized, the primals are feasible, and complementary slackness holds. Under this optimal solution to $\ms L(\infty,\infty)$, both the buyer queue and the goods queue have lengths at most $K$.
\end{proof}

\section{Why Opacity and Passivity Are Necessary for Dominance}\label[onlineapp]{sec:appendix:reason-for-passive-buyers}

We provide examples showing that the restrictions imposed on the Cutoff-Price Mechanism are necessary for the mechanism to be strategyproof.

\subsection{Hiding Buyer Identities and Past Bids}

Suppose the queue starts empty, and the order of buyer and item arrivals is as shown in \cref{figure:passive-buyers-example}. Suppose the thresholds are $\hat v_1 = 1$, $\hat v_2 = 2$, and $\hat v_3 = 4$. 

\begin{figure}[htb]
\centering
\begin{tikzpicture}[>=stealth, thick, xscale=1.1]
    \draw[->] (-0.5,3) -- (9,3) node[right] {};
    \foreach \x/\label in {0/$v_A=3$, 1.5/$v_B=6$, 3/item, 4.5/$v_C=5$, 6/$v_D=5$, 7.5/item} {
        \draw (\x, 3.1) -- (\x, 2.9) node[below] {\small \label};
    }
\end{tikzpicture}
\caption{}
\label{figure:passive-buyers-example}
\end{figure}

Suppose that, unlike in our definition of the Cutoff-Price Mechanism, buyers in the queue observe each other's identities. Suppose $B$ bids truthfully, whereas $C$ and $D$ both adopt the strategy under which they do not enter the queue (or, equivalently, drop out immediately in their first survival auction) if they see $B$ in the queue, and otherwise behave truthfully. Against these strategies, if $A$ behaves truthfully, $B$ wins the first item, and $A$ drops out when buyer $D$ arrives. A profitable deviation for $A$ is to bid 7 in the first assignment auction and win. Then, $C$ and $D$ do not enter, so $B$ wins the second item. $A$'s cutoff price is 2, since taking the buyers' realized behavior as given, $A$ would have won an item as long as he survived until the second assignment auction.\footnote{Of course, if $A$ had actually lost the first assignment auction, then $C$ and $D$ would have behaved truthfully, so $A$ would not have won the second item. However, the mechanism does not observe this ``true'' counterfactual. From the mechanism's point of view, it is as if $v_A=7$, $v_B=6$, $v_C=v_D=0$, and everyone bid truthfully, so the cutoff price must equal 2.}

Likewise, suppose the buyers observe past bids. Then, if $B$ bids truthfully, while $C$ and $D$'s strategy is to not enter if the winning bid of the first assignment auction was 7 and otherwise behave truthfully, then $A$ is again better off bidding 7 than being truthful.

Under the rules of CPM, buyers' identities or past bids are not disclosed. Therefore, $C$ and $D$'s strategies cannot condition on such information.

\subsection{Buyers Becoming Passive, and Hiding Queue Length}

Suppose the queue starts empty, and the order of buyer and item arrivals is as shown in \cref{figure:passive-buyers-example}. Suppose the thresholds are $\hat v_1 = 1$, $\hat v_2 = 2$, $\hat v_3 = 3$, and $\hat v_4 = 9$.
\medskip
\begin{figure}[htb]
\centering
\begin{tikzpicture}[>=stealth, thick, xscale=1.1]
    \draw[->] (-0.5,3) -- (12,3) node[right] {};
    \foreach \x/\label in {0/$v_A=4$, 1.5/$v_B=6$, 3/item, 4.5/$v_C=5$, 6/item, 7.5/$v_D=10$, 9/$v_E=10$, 10.5/$v_F=10$} {
        \draw (\x, 3.1) -- (\x, 2.9) node[below] {\small \label};
    }
\end{tikzpicture}
\caption{}
\label{figure:passive-buyers-example}
\end{figure}

Suppose that, unlike in our definition of CPM, buyers remain active even after playing an assignment auction. Suppose $B$'s strategy is to behave truthfully until the end of the first assignment auction, and if he loses this auction, he drops out in the next survival auction at the clock price of $\hat v_1=1$.  Moreover, suppose $C$'s strategy is to behave truthfully, except that if his first survival auction ends with clock price $\hat v_1$, he subsequently behaves as if he has value $\hat v_1$. Suppose $D$, $E$, and $F$ behave truthfully.
Against these strategies, if $A$ behaves truthfully, he drops out when $F$ arrives. One profitable deviation for $A$ is to bid 7 in the first assignment auction and win. This leads $B$ to drop out when $C$ arrives, which in turn leads $C$ to bid $\hat v_1=1$ in the second assignment auction. %
$A$'s cutoff price is $\hat v_2 =2$, since taking the buyers' realized behavior as given, as long as $A$ bids no less than 2 for the first item, he will survive when $C$ arrives and then win the second item against $C$'s bid of $\hat v_1=1$.

This example shows that, even if information about buyer identities and past bids are hidden, dominance fails if buyers do not become passive, since buyers necessarily observe results of auctions that they play and may condition their future behavior on these results. Under the rules of CPM, buyers become passive once they participate in an assignment auction. This guarantees dominance by limiting buyers' ability to condition their play on past auction results. For instance, in the example above, if $B$ bids $b\ge \hat v_2=2$  in the first assignment auction, $B$'s proxy agent will not drop out in a subsequent survival auction at the clock price of $\hat v_1=1$.

Of course, we assume that buyers are free to leave the queue at any time, so even under CPM, $B$ is technically able to drop out after becoming passive. This is why CPM requires that the proxy agent for $B$ continue to play according to the fixed bid that $B$ submitted for the first item, even after the actual buyer drops out. 

Finally, the above example can be altered slightly to show why CPM must hide the length of the  queue. Suppose that buyers become passive as in CPM, but queue lengths are disclosed to the buyers. Suppose $B$'s strategy is to behave truthfully until the end of the first assignment auction, and if he loses this auction, he drops out immediately. Note that this is feasible for $B$, even under CPM. Suppose $C$'s strategy is to behave truthfully, except that if the true queue length (excluding the number of proxy agents) is 0 when he arrives, he behaves as if he has value $\hat v_1$. Then, $A$ can once again profitably deviate from truth-telling by bidding 7 for the first item.

\end{document}